\documentclass[11pt]{article}

\usepackage[T1]{fontenc}
\usepackage[utf8]{inputenc}
\usepackage{lmodern}
\usepackage{microtype}
\usepackage[a4paper,margin=31mm,headheight=15pt]{geometry}
\usepackage{amsmath,amssymb,amsthm,mathtools,bm}
\usepackage{booktabs,tabularx,array}
\usepackage{enumitem}
\usepackage{xcolor}
\usepackage{fancyhdr}
\usepackage[round,authoryear]{natbib}
\usepackage{hyperref}

\definecolor{ink}{HTML}{17253A}
\definecolor{linkblue}{HTML}{2F5D7C}
\hypersetup{
  colorlinks=true,
  linkcolor=linkblue,
  citecolor=linkblue,
  urlcolor=linkblue,
  pdftitle={Anonymous Accountability},
  pdfauthor={Georgy Lukyanov and Muslim Salakhov}
}

\newtheorem{theorem}{Theorem}[section]
\newtheorem{proposition}[theorem]{Proposition}
\newtheorem{lemma}[theorem]{Lemma}
\newtheorem{corollary}[theorem]{Corollary}
\theoremstyle{definition}

\theoremstyle{remark}
\newtheorem{remark}[theorem]{Remark}

\newcommand{\E}{\mathbb E}
\newcommand{\Prb}{\mathbb P}
\newcommand{\1}{\mathbf 1}
\newcommand{\Hid}{\mathsf H}
\newcommand{\Comp}{\mathsf C}
\newcommand{\Ind}{\mathsf I}
\newcommand{\Pois}{\mathrm{Pois}}
\newcommand{\Bin}{\mathrm{Bin}}
\newcommand{\dd}{\,\mathrm d}
\newcommand{\eps}{\varepsilon}
\allowdisplaybreaks

\title{\textbf{Anonymous Accountability}}
\author{Georgy Lukyanov\thanks{Corresponding author. Toulouse School of Economics, 1 Esplanade de l'Universit\'e, 31080 Toulouse Cedex 06, France. Email: \href{mailto:georgy.lukyanov@tse-fr.eu}{georgy.lukyanov@tse-fr.eu}.}
\and Muslim Salakhov\thanks{HSE University, International College of Economics and Finance, 11 Pokrovsky Boulevard, Moscow 109028, Russia. Email: \href{mailto:mdsalakhov@edu.hse.ru}{mdsalakhov@edu.hse.ru}.}}
\date{September 2026}

\begin{document}
\maketitle

\begin{abstract}
We ask how an expert panel should disclose auditable conflicts of interest when reports are named and experts care about their reputations for ability. Before sealed binary assessments are filed, an auditor commits to disclosing no orientation information, each named expert's orientation, or only the number of experts oriented in each direction; disclosure occurs after reports are locked. Holding reporting behaviour fixed, composition disclosure and individual attribution are exactly Blackwell-equivalent for learning the state at every finite panel size. Conditional on named reports remaining public, composition is minimal among state-lossless disclosures of orientations: every lossless public message must make the composition recoverable together with the reports. Incentives nevertheless differ sharply. In the large-panel limit near pooling, the local career returns to following evidence against one's orientation under attribution, secrecy, and composition disclosure stand in the ratio $1:4:5$. On an open set of parameters, the two extreme regimes have only pooling within symmetric cutoff strategies, while composition disclosure admits a tremble-robust informative equilibrium. Near the lower boundary of this region, a sequence of composition equilibria is sparse---the evidence-following rate is of order $1/N$---yet induces a nondegenerate state experiment. Anonymous disclosure can therefore preserve the information available under attribution while creating stronger incentives to produce it.

\end{abstract}

\clearpage

\section{Introduction}

An advisory committee votes on a drug application, and several of its members have consulted for the sponsor; a scientific panel is asked whether a chemical is safe, and part of the panel is funded by the industry that makes it; an investment committee ranks a deal, and its members sit on that committee precisely because they represent different mandates; a forecasting group files directional calls, and every analyst belongs to an institution with a well-known house view. Many such bodies require declarations of interests, although their disclosure and verification rules vary. Taking an auditable register as given, we ask a separate design question: at which \emph{granularity} should its contents be published?

Consider a regulator that asks a panel of named experts to file sealed assessments of a binary risk. Each expert has an auditable orientation---a financial tie, an institutional mandate, a methodological commitment, or any other persistent characteristic that predictably favours one conclusion. Before reports are filed, the auditor commits to publish nothing about these orientations, identify the orientation of every expert, or publish only how many experts point in each direction once reports are locked. The state is eventually verified, so a report leaves a mark on its author's reputation for ability. Which of these policies elicits the most evidence, and how much information about the state has to be given up in order to obtain it?

The usual comparison is between the two extremes, secrecy and attribution. The policy studied here sits between them: disclose the panel's \emph{composition} without disclosing the identity-to-orientation mapping. It is anonymous with respect to conflicts only; reports remain named and every expert remains individually accountable for her accuracy. It is shown that this limited anonymity separates two functions that transparency ordinarily bundles together. Aggregate disclosure preserves everything an outsider needs in order to interpret the panel, while withholding exactly the assignment that would allow the market to excuse a particular expert's orientation-consistent report.\footnote{The two functions are usually discussed as one, under the single heading of accountability. They are logically distinct: an audience that wants to \emph{read} a panel needs to know what the panel is made of, whereas an audience that wants to \emph{forgive} a report needs to know who made it and why. Our disclosure rule supplies the first and withholds the second.}

The environment has $N$ experts, an unknown binary state, and a persistent high or low signal ability. Experts do not know their own ability, so a report is a bet on one's own competence as much as on the state.\footnote{This is the standard reputational-cheap-talk formulation of \citet{ottaviani2006professional,ottaviani2006reputational}. If experts knew their types, reporting would generally use ability-specific cost cutoffs, and the common-fidelity career return compared across regimes would no longer be the same object.} Each expert privately observes a signal, an auditable directional orientation, and a private cost of reporting against that orientation. Reports are simultaneous, mandatory, named, and irrevocable. The disclosure regime is publicly committed in advance and implemented only after reports are locked. Once the state is verified, the market updates every expert's reputation for ability. An expert who receives evidence against her own orientation therefore trades off the private cost of filing that evidence against the induced increase in her expected reputation.

Three results carry the argument. First, holding reporting behaviour fixed, composition disclosure and individual attribution are \emph{exactly Blackwell-equivalent} for state inference at every finite $N$. In one direction this is obvious: the orientation vector can simply be erased. In the other direction, conditional on the composition and named reports, the missing assignment of conflicts to identities can be reconstructed by a randomisation that does not depend on the state. Conditional on named reports remaining public, every state-lossless garbling must make the composition recoverable from its message together with those reports. Composition is therefore minimal in the refinement order among lossless disclosures of orientations. A finite-alphabet extension shows that additive likelihood separability is always sufficient for histogram disclosure to be lossless and, under full-support orientation uncertainty, is also necessary.

Second, the incentive comparison reverses the informational equivalence. Let $x$ denote the probability that an expert follows her signal when it conflicts with her orientation. At pooling, the large-panel derivatives of the reputational return satisfy
\begin{equation}
  \Delta'_{\Ind}(0):\Delta'_{\Hid}(0):
  \lim_{N\to\infty}\Delta'_{\Comp,N}(0)=1:4:5.
  \label{eq:intro-ratio}
\end{equation}
Individual attribution produces the ``1.'' Once a conflict has been attributed, an orientation-consistent report is discounted as evidence, but it is also excused as the predictable consequence of a known bias; locally, the excuse dominates. Secrecy removes the excuse and produces the ``4.'' Composition disclosure then adds something that neither extreme can generate---a genuinely collective effect. Because the market knows how many orientations of each sign exist but not who holds them, a rare orientation-defying report shifts the posterior assignment of conflicts across the entire named report profile. This residual cross-identification produces the ``5.'' The three large-panel coefficients are pure numbers: every signal and ability primitive of the environment cancels.\footnote{All three slopes are proportional to the same constant $\chi=4md$ defined in \eqref{eq:mchi}, which collects the prior, the gap between the two abilities, and the average informativeness of a signal. The ratio in \eqref{eq:intro-ratio} is therefore invariant to all of them. The finite-panel composition coefficient still depends on $N$.}

That local difference creates an equilibrium that neither extreme regime supports. With uniformly distributed reporting costs, let $\kappa$ summarise career stakes, cost density, and the informativeness of ability. On the open parameter region
\[
  \frac45<\kappa<1-d^2,
\]
where $d$ is average signal informativeness, secrecy and individual attribution have only pooling within the regular strategy family, while composition disclosure has an interior regular equilibrium for every panel at or above an explicit sufficient threshold. Pooling still coexists in the unperturbed game. For each fixed qualifying panel, a vanishing forced-fidelity perturbation sends every extreme-regime equilibrium to pooling while every composition equilibrium remains bounded away from it. The resulting decision-accuracy comparison therefore does not rely on choosing one favourable composition root.

Third, the branch emerging near $\kappa=4/5$ is sparse and yet consequential. It satisfies $x_N\sim c^*/N$, so the number of orientation-defying reports stays of order one even as the panel grows without bound. The limiting discrepancy between positive reports and positive orientations is Skellam distributed, and the observer's state-classification accuracy remains strictly above one half. At the same $1/N$ fidelity under secrecy, by contrast, the two state-contingent report distributions merge and accuracy converges to one half. What matters for aggregate learning is the \emph{number} of candid reports rather than their share. A large panel makes this distinction stark, and disclosed composition is precisely what makes the handful legible.\footnote{The mechanism behind the contrast is worth stating plainly. Under secrecy, a handful of candid reports is diluted by the $N-O(1)$ reports that merely repeat orientations, and the dilution is fatal because the observer does not know what the orientations were. Under composition disclosure the observer knows the aggregate benchmark and can therefore read the same handful of reports as a deviation from it. The candid reports are not more numerous; they are merely no longer invisible.} A local robustness result shows that the composition mechanism can also generate first-order accountability when the state is never verified, although the baseline equilibrium results use eventual verification.

These results suggest a practical institutional design. An advisory body could require sealed named assessments, maintain an audited register of directional conflicts, and release an issue-level conflict count only once the assessments are locked. The natural applications are scientific advisory panels whose members have industry or advocacy ties, panels evaluating contested technologies, investment committees whose members represent different mandates, and forecasting teams with auditable institutional or model-class affiliations. We do not claim that any existing institution implements exactly this policy. Eventually scored settings fit the baseline most literally, while the verification-free result covers environments in which peers must supply the benchmark. The timing remains essential: releasing the realised composition \emph{before} filing would permit composition-contingent strategies and is a different design problem altogether.

\subsection{Relation to the literature}

The paper belongs first to the literature on reputational communication. Career concerns are known to distort professional advice even when the state is later observed \citep{ottaviani2006professional,ottaviani2006reputational}; they can induce committees to manipulate information or to present a united front \citep{visser2007committees}; and they can make transparency counterproductive by encouraging conformism \citep{prat2005wrong,levy2007committees}. \citet{fehrler2018transparency} show theoretically and experimentally how transparency changes information aggregation in a career-concerned committee. Their transparency regimes concern deliberation and votes; here all reports remain named and the designed object is the ex post assignment of auditable orientations. Recent work continues to show that reputational motives interact sharply with what the audience knows and with the precision of public communication \citep{balmaceda2021private,mattozzi2023public,gati2025confidence}. Our contribution is a disclosure theorem for a many-expert environment in which the reputational object is ability while the disclosed object is an auditable directional conflict.\footnote{The distinction matters for the comparison with \citet{prat2005wrong}. There, transparency of \emph{actions} rather than of \emph{consequences} is what harms the principal, because an agent who is observed acting will imitate the action a good type is expected to take. Here every action is observed in all three regimes; what varies is how much the audience knows about the private inducements attached to it.}

The closest antecedent is \citet{bourjade2011roles}. They study biased experts who value reputation and show, among other results, that adding a second expert has opposite effects on truth-telling depending on whether individual or only collective contributions are identified. Our object and our mechanism are different. Reports are always individually identified in our model; what may remain anonymous is the mapping from experts to conflicts. This permits an exact comparison of secrecy, composition, and attribution for arbitrary $N$, and it is what produces the finite-panel multiplier, the $1:4:5$ limit, and the exact state-information equivalence between composition and full attribution. The last of these is what allows the paper to separate information about the state from information about responsibility.

The analysis is also related to work on the disclosure of expert bias. Mandatory disclosure need not improve cheap-talk communication \citep{li2008mandatory}, and the interaction between reputation for ability and conflicts of interest can make truthful revelation harder rather than easier \citep{pavesi2014experts}. \citet{levit2022recommendations} show how a common recommendation can conceal an expert's privately known conflict, while \citet{mezzetti2025manipulative} studies verifiable disclosure when the direction of expert bias is unknown. \citet{li2025transparency} shows that transparency about preferences changes the endogenous provision of information to a reputation-minded policymaker. Related work studies how public voting or group-disclosure procedures allocate social pressure and blame \citep{namecorrea2019social,onuchic2025groups}. In those papers votes or outcomes are themselves selectively public; here an outside auditor commits to a conflict disclosure rule, all outcome reports remain public, and what may stay anonymous is the assignment of pre-existing orientations. The conclusion is not that less disclosure is intrinsically better---a claim we would not defend---but that one particular statistic can be sufficient for the state and deliberately insufficient for attribution.

A further strand asks how the agendas, information sources, and voting institutions of expert panels should be composed or disclosed \citep{krishna2001model,gersbach2012information,bhattacharya2018optimality,catonini2024expertise}. \citet{bhattacharya2018optimality}, for instance, show that the desirability of preference diversity depends on how experts' information is correlated, while \citet{catonini2024expertise} compare independent and collective expertise when experts may exchange information before advising. We take the opposite margin: experts never deliberate or observe one another before reporting, the ex ante orientation distribution is held fixed, and what is designed is the revelation of its realised assignment. The relevant panel force is cross-identification, which is absent with a single expert and strengthens with panel size.

The verification-free extension touches the literature on eliciting information without ground truth and peer prediction \citep{prelec2004truth,miller2005peer}. Classical mechanisms in that literature design transfers around the correlation of peer reports. We design no transfer and do not establish global truthful implementation without verification: reports remain named, experts' payoffs increase with posterior reputation for latent ability, and an audited composition count creates the endogenous benchmark. Our result on this margin is correspondingly local.

Finally, \citet{sethi2016communication,sethi2026culture} distinguish the quality of a source's information from how readily its perspective can be interpreted. That distinction motivates our separation of ability and orientation, but the strategic object is different: orientations are auditable, reports are reputationally motivated, and an auditor chooses what orientation information to release. The model is also distinct from \citet{lukyanov2026silence}, which studies a two-expert, two-stage protocol with interim feedback, revision, and silence. Here reports are simultaneous and irrevocable, there is no peer feedback and no option to stay quiet, and every comparative result is generated by ex post conflict attribution in an $N$-expert panel. Neither the disclosure object nor the policy space in one paper nests its counterpart in the other.

Section~\ref{sec:model} presents the environment. Section~\ref{sec:benchmarks} derives the two extreme regimes. Section~\ref{sec:composition} establishes the finite-panel incentive theorem and Section~\ref{sec:equilibrium} the equilibrium-creation result. Section~\ref{sec:information} proves exact information equivalence, minimality, and the finite-alphabet characterisation. Section~\ref{sec:sparse} studies the large-panel equilibrium and the state information it generates. Section~\ref{sec:discussion} discusses implementation, verification-free peer validation, and robustness. Proofs are collected in the Appendix.

\section{Environment}\label{sec:model}

The model is built to make one comparison as clean as possible: the same experts, the same reports, the same verification, and three different answers to the question of what the public is told about the panel's conflicts. Everything that is not needed for that comparison has been kept deliberately thin.

\subsection{State, ability, evidence, and orientation}

There are $N\geq1$ named experts. The state is $\theta\in\{-1,+1\}$, with prior probability one half on each realisation. Expert $i$ has a fixed latent ability
\[
  \tau_i\in\{H,L\},\qquad \Prb(\tau_i=H)=p\in(0,1),
\]
independently across experts. Neither the expert nor the public observes $\tau_i$. Conditional on $(\theta,\tau_i=t)$, expert $i$ observes $s_i\in\{-1,+1\}$ with
\[
  \Prb(s_i=\theta\mid\tau_i=t)=q_t,
  \qquad \frac12<q_L<q_H<1.
\]
Write
\begin{equation}
 h=q_H,\quad \ell=q_L,\quad Q=ph+(1-p)\ell,\quad
 d=2Q-1,
 \label{eq:primitives}
\end{equation}
and define
\begin{equation}
 m=p(1-p)(h-\ell),\qquad \chi=4md.
 \label{eq:mchi}
\end{equation}
Because an expert does not know her own ability, her posterior probability that her own signal is correct is $Q$, the population average of $q_H$ and $q_L$. The quantity $m$ is the covariance-like term that governs how strongly a single correct report moves the market's belief about ability, and $\chi=4md$ is the reputational stake attached to being right in the secrecy benchmark; both will scale every incentive computed below.

Expert $i$ also observes an issue-specific orientation $b_i\in\{-1,+1\}$. Orientations are i.i.d., independent of all other primitives, and symmetric:
\[
 \Prb(b_i=+1)=\Prb(b_i=-1)=\frac12.
\]
An auditor can verify orientations, but the public does not initially know the identity-to-orientation mapping. The model therefore presumes that the auditor's named classifications are confidential before release; if outsiders can already infer every orientation from public affiliations, publishing only a count cannot restore anonymity. An orientation may stand for a sponsor's preferred stance, an institutional mandate, or any persistent auditable characteristic that makes one report privately easier to file than the other. Expert $i$ privately observes the cost $v_i\geq0$ of reporting against $b_i$; costs are i.i.d.\ with cdf $F$ and independent of everything else. Thus $v_i$ is a reduced-form cost of publicly departing from a stance---not a payoff from the observer's ultimate action. The symmetric i.i.d.\ orientation prior is a benchmark restriction used to obtain the scalar incentive formula and the $1:4:5$ comparison. Corollary~\ref{cor:orientationlaw} shows that the state-information equivalence itself requires neither symmetry nor independence across orientations.

\subsection{Reports, disclosure, and reputation}

The timing is as follows.
\begin{enumerate}[leftmargin=2em,itemsep=.25em]
 \item The auditor publicly commits to a disclosure regime $D\in\{\Hid,\Ind,\Comp\}$.
 \item Nature draws the state, types, signals, orientations, and reporting costs.
 \item Experts simultaneously file named and irrevocable reports $r_i\in\{-1,+1\}$. Filing is mandatory. Experts do not observe the realised panel composition before reporting.
 \item After reports are locked, the auditor releases the information prescribed by the committed regime:
 \begin{itemize}[leftmargin=1.5em,itemsep=.1em]
  \item $\Hid$ (hidden): no orientation information is disclosed;
  \item $\Ind$ (individual attribution): the named vector $(b_1,\ldots,b_N)$ is disclosed;
  \item $\Comp$ (composition): only $K=\sum_i\1\{b_i=+1\}$ is disclosed.
 \end{itemize}
 \item An observer sees the public reports and disclosed data and chooses $a\in\{-1,+1\}$, receiving payoff $\1\{a=\theta\}$.
 \item The state is publicly verified. A reputation market sees the same public history and the verified state, then updates each expert's reputation, measured by the posterior probability that her type is $H$.
\end{enumerate}

Let $\mu_i^D$ denote expert $i$'s posterior reputation after verification under $D\in\{\Hid,\Ind,\Comp\}$. Expert $i$ maximises the reporting-stage expectation of
\begin{equation}
 u_i=\rho\mu_i^D-v_i\1\{r_i\neq b_i\}
 -\Phi\1\{r_i\notin\{s_i,b_i\}\},
 \qquad \Phi>\rho>0.
 \label{eq:utility}
\end{equation}
The last term is a no-fabrication discipline, and it is deliberately inert. Since reputation lies in $[0,1]$, an expert whose signal agrees with her orientation never reports against both; and when $s_i=-b_i$, both available reports belong to $\{s_i,b_i\}$, so the term never enters an incentive comparison.\footnote{The device rules out the uninteresting equilibria in which an expert files a report supported by neither her evidence nor her interest---reports that would be off path in any case, but whose off-path reputations would otherwise have to be specified.} The observer's action does not enter the expert's baseline payoff; Section~\ref{sec:discussion} returns to this restriction.

\subsection{Symmetric fidelity and equilibrium}

Throughout we study symmetric monotone strategies indexed by a single number, the fidelity rate $x\in[0,1]$:
\begin{equation}
 r_i=\begin{cases}
 s_i,&s_i=b_i,\\
 s_i,&s_i=-b_i\text{ and }v_i\leq \bar v,\\
 b_i,&s_i=-b_i\text{ and }v_i>\bar v.
 \end{cases}
 \label{eq:strategy}
\end{equation}
When evidence and interest agree there is nothing to decide; when they disagree, the expert follows the evidence if and only if the private cost of doing so is below a cutoff. Thus $x=F(\bar v)$ is the probability of following evidence in a conflict, and it is the only behavioural object in the paper.\footnote{The restriction to this one-parameter family is what makes the three regimes comparable at all: the reputational return $\Delta_{D,N}$ is then a function of a scalar, and the equilibrium condition is a scalar fixed point rather than a system. It is also without loss within the class of symmetric cutoff strategies, since any such strategy is described by its cutoff.} Conditional on $s_i=-b_i$, define the expected reputational benefit of reporting the signal rather than the orientation by
\begin{equation}
 \Delta_{D,N}(x)=\E\!\left[
 \mu_i^D(r_i=s_i)-\mu_i^D(r_i=b_i)\mid s_i=-b_i
 \right],
 \label{eq:delta}
\end{equation}
where peers use fidelity $x$, the focal report is intervened on, and the state is subsequently observed. Symmetry makes the expression independent of the sign of $b_i$. A symmetric equilibrium is then a solution of
\begin{equation}
 x=T_{D,N}(x):=F\!\left(\rho\Delta_{D,N}(x)\right).
 \label{eq:fixedpoint}
\end{equation}
At $x=0$ an orientation-defying report is off path, and its reputational consequences are therefore not pinned down by Bayes' rule. We define $\Delta_{D,N}(0)$ and its right derivative as limits along the symmetric strategy family as $x\downarrow0$; these are the \emph{regular posteriors}. The convention is equivalent to a vanishing symmetric behavioural tremble. Pooling survives under composition disclosure alongside the informative equilibrium, and Section~\ref{subsec:selection} studies the persistence of the positive branch under an explicit perturbation without claiming unique selection.

\begin{remark}
The regime is committed ex ante, but its prescribed orientation information is released only after reports are locked. Composition disclosure is therefore not an interim public signal to the experts. If experts observed the realised $K$ before filing, equilibrium fidelity would depend on $K$ and the scalar fixed point \eqref{eq:fixedpoint} would no longer describe the game. The sealed-report timing is also what separates this design from models of sequential deliberation, feedback, and revision: nobody here learns anything from anybody before speaking.
\end{remark}

\section{The two extreme regimes}\label{sec:benchmarks}

The two extreme regimes admit closed-form reputational returns, and they are worth deriving in full: they are the benchmark against which the panel effect of composition disclosure will be measured, and the contrast between them already contains the paper's central tension in miniature.

\subsection{Hidden orientations}

Under $\Hid$, a type-$t$ expert reports the verified state with probability
\[
 \alpha_t(x)=\frac{1+x(2q_t-1)}2.
\]
Bayes' rule gives the posterior reputations after a correct and an incorrect report:
\[
 \mu_{\Hid}^{+}(x)=\frac{p[1+x(2h-1)]}{1+xd},
 \qquad
 \mu_{\Hid}^{-}(x)=\frac{p[1-x(2h-1)]}{1-xd}.
\]
Their difference is $4mx/(1-d^2x^2)$, which is increasing in $x$: the more the panel is expected to follow its evidence, the more a correct report is worth. Conditional on a conflicting signal, the state agrees with the signal with probability $Q$ and with the orientation with probability $1-Q$.

\begin{proposition}\label{prop:hidden}
For every $x\in[0,1]$,
\begin{equation}
 \Delta_{\Hid}(x)=\frac{\chi x}{1-d^2x^2}.
 \label{eq:deltaH}
\end{equation}
In particular, $\Delta'_{\Hid}(0)=\chi$.
\end{proposition}

\subsection{Individual attribution}

Under $\Ind$ the market knows whether the focal expert's orientation was aligned with the realised state, and it conditions on that fact. If the state agrees with the signal, the posterior gain from having followed it is
\[
 \frac{ph}{Q}-\frac{p(1-xh)}{1-xQ}.
\]
If the state agrees with the orientation, the corresponding posterior difference is
\[
 \frac{p(1-h)}{1-Q}-
 \frac{p[1-x(1-h)]}{1-x(1-Q)}.
\]
Weighting these two cases by $Q$ and $1-Q$ yields the following formula.

\begin{proposition}\label{prop:individual}
For every $x\in[0,1]$,
\begin{equation}
 \Delta_{\Ind}(x)=
 \frac{\chi x}{(2-x)^2-d^2x^2}.
 \label{eq:deltaI}
\end{equation}
Consequently,
\[
 \Delta'_{\Ind}(0)=\frac{\chi}{4},
 \qquad
 \Delta_{\Ind}(1)=\Delta_{\Hid}(1)=\frac{\chi}{1-d^2}.
\]
\end{proposition}

Full attribution therefore reduces the local career return to orientation-defying evidence by a factor of four relative to secrecy. This is the excuse effect, and it is worth stating carefully, because common sense would suggest the opposite: knowing an expert's conflict ought, one might think, to make her orientation-consistent report less persuasive and hence more costly to file. It does make it less persuasive. But the same knowledge also \emph{explains} the report---an expert who says what her interest dictates is doing what a biased expert of any ability would do, and she is therefore judged only lightly on it. Locally, the second effect dominates the first, and the net result is that attribution protects the conflicted expert from the reputational consequences of following her interest. We use this one-expert force as a benchmark rather than claim it as the paper's novelty; the new effect begins when an anonymous composition constraint links several named reports.

\section{Anonymous composition at finite \texorpdfstring{$N$}{N}}\label{sec:composition}

Composition disclosure does not identify which experts carry which orientations, and this is where the panel starts to matter. A focal reputation must now be computed jointly from the focal report, the named peer reports, and the aggregate constraint $K$---the market has to allocate a known number of conflicts across a known profile of reports, and every report it looks at bears on how it allocates the rest. We first give an exact representation of that posterior, and then extract its local incentive effect.

For $q\in(1/2,1)$, orientation $b$, and state $\theta$, let $P_q^b(r\mid\theta,x)$ be the report probability implied by \eqref{eq:strategy}. It satisfies
\begin{align*}
 P_q^{\theta}(\theta\mid\theta,x)&=1-x(1-q),&
 P_q^{\theta}(-\theta\mid\theta,x)&=x(1-q),\\
 P_q^{-\theta}(\theta\mid\theta,x)&=xq,&
 P_q^{-\theta}(-\theta\mid\theta,x)&=1-xq.
\end{align*}
For $n$ peers, define the bivariate generating polynomial
\begin{equation}
 G_\theta(z,w;x)=\frac12\sum_{b\in\{-1,+1\}}
 z^{\1\{b=+1\}}
 \left[1-P_Q^b(+1\mid\theta,x)+P_Q^b(+1\mid\theta,x)w\right]
 \label{eq:G}
\end{equation}
and
\begin{equation}
 P_{n,\theta}(k,a;x)=[z^kw^a]G_\theta(z,w;x)^n.
 \label{eq:coefficient}
\end{equation}
This is the joint probability that $k$ peers hold positive orientations and $a$ peers file positive reports.

For focal type $t$, focal report $r$, total composition $K$, and $a$ positive peer reports, let
\begin{align}
 \Lambda_t(r,K,a\mid\theta,x)
 ={}&\frac12P_{q_t}^{-1}(r\mid\theta,x)P_{N-1,\theta}(K,a;x)\nonumber\\
 &+\frac12P_{q_t}^{+1}(r\mid\theta,x)P_{N-1,\theta}(K-1,a;x).
 \label{eq:focalL}
\end{align}
The focal posterior is
\begin{equation}
 \mu(r,K,a,\theta;x)=
 \frac{p\Lambda_H(r,K,a\mid\theta,x)}
 {p\Lambda_H(r,K,a\mid\theta,x)+(1-p)\Lambda_L(r,K,a\mid\theta,x)}.
 \label{eq:muC}
\end{equation}
Exchangeability makes the number of positive peer reports sufficient for the focal posterior even though the reports themselves are named---which identity filed which report is irrelevant once the count is known.\footnote{This is a convenience, not a substantive restriction. It fails as soon as experts are heterogeneous in ability distribution or in orientation prior, in which case the focal posterior depends on the full named profile. The likelihood representation generalises, but the scalar reduction and the multiplier in Theorem~\ref{thm:multiplier} need not.} Fixing $b_i=+1$ and $s_i=-1$, the exact finite-panel incentive is
\begin{align}
 \Delta_{\Comp,N}(x)
 ={}&\sum_{\theta\in\{-1,+1\}}\Prb(\theta\mid s_i=-1)
 \sum_{k=0}^{N-1}\sum_{a=0}^{N-1}P_{N-1,\theta}(k,a;x)\nonumber\\
 &\quad\times\left[
 \mu(-1,k+1,a,\theta;x)-\mu(+1,k+1,a,\theta;x)
 \right].
 \label{eq:deltaCexact}
\end{align}
Equation \eqref{eq:deltaCexact} is a finite sum of polynomial-rational terms. Its level is cumbersome and not especially informative. Its derivative at pooling, on the other hand, is remarkably simple.

\begin{theorem}\label{thm:multiplier}
Under regular posteriors,
\begin{equation}
 \Delta'_{\Comp,N}(0)=g_N\Delta'_{\Hid}(0),
 \qquad
 g_N=\frac54-\frac{5}{2N}+\frac{N+5}{N2^{N+1}}.
 \label{eq:gN}
\end{equation}
Moreover, $g_N$ is strictly increasing, $g_1=1/4$,
\[
 g_N<1\quad(N\leq9),\qquad
 g_N>1\quad(N\geq10),\qquad
 \lim_{N\to\infty}g_N=\frac54.
\]
Hence the large-panel local slopes satisfy
\begin{equation}
 \Delta'_{\Ind}(0):\Delta'_{\Hid}(0):
 \lim_{N\to\infty}\Delta'_{\Comp,N}(0)=1:4:5.
 \label{eq:145}
\end{equation}
\end{theorem}

The proof, in Appendix~\ref{app:finite}, conditions on the number of peer orientations aligned with the verified state and carries out a complete one-departure expansion. It reduces the derivative to three binomial expectations, each of which is evaluated in closed form. No numerical computation enters \eqref{eq:gN} at any point.

The theorem makes the panel nature of the mechanism transparent. With $N=1$ the composition \emph{is} the orientation, so composition and individual attribution coincide and $g_1=1/4$, exactly the attribution value. With many experts the disclosed count couples every report to the possible assignments of orientations across the peers. When one expert files a rare orientation-defying report, the market must decide whether the resulting discrepancy is focal evidence or the trace of a different assignment of the same conflicts to the same profile of reports. That alternative-assignment term is the residual cross-identification effect, and its limiting contribution lifts the secrecy benchmark from four normalised units to five.

The comparison is therefore not monotone in disclosure, which is the least intuitive part of the result. Small panels contain too little cross-identification to offset the excuse that partial attribution creates, so composition disclosure first overtakes individual attribution and only later, from ten experts onward, overtakes secrecy as well. The cutoff $N=10$ and the increasing, single-crossing shape are exact properties of the symmetric binary benchmark, not universal empirical thresholds. A regulator choosing between these regimes for a three-person committee and for a fifty-person panel is not facing the same problem.

\section{Equilibrium creation}\label{sec:equilibrium}

To translate the local incentive theorem into equilibrium behaviour we need a cost distribution. Suppose
\begin{equation}
 v_i\sim\mathrm{Unif}[0,1/f],\qquad
 F(v)=\begin{cases}0,&v<0,\\ \min\{fv,1\},&v\geq0,
 \end{cases}
 \label{eq:uniform}
\end{equation}
and define the dimensionless career parameter
\begin{equation}
 \kappa=\rho f\chi.
 \label{eq:kappa}
\end{equation}
Condition~\eqref{eq:openregion} below keeps the extreme-regime cutoffs and the composition value at $x=1$ in the linear part of the cdf. The sparse composition arguments likewise use unsaturated cutoffs because their arguments are $O(N^{-1})$. Only the density of costs near zero matters for the local argument; the uniform specification is a convenience that also delivers the global comparisons.\footnote{The parameter $\kappa=\rho f\chi$ is the product of the weight $\rho$ that experts place on reputation, the density $f$ of reporting costs near zero, and the reputational stake $\chi$. It is the one dimensionless number that decides whether evidence is worth filing at all, and the region \eqref{eq:openregion} below requires it to be large---careers must matter and the marginal cost of candour must be small---but not so large that candour would be automatic under any regime.}

\begin{proposition}\label{prop:pooling}
If $\kappa<1-d^2$, then $x=0$ is the only symmetric fixed point under $\Hid$ and under $\Ind$.
\end{proposition}

\begin{proof}
For $x>0$, the raw linear indices from Propositions~\ref{prop:hidden}--\ref{prop:individual} satisfy
\[
 f\rho\Delta_{\Hid}(x)=\frac{\kappa x}{1-d^2x^2},
 \qquad
 f\rho\Delta_{\Ind}(x)=
 \frac{\kappa x}{(2-x)^2-d^2x^2}.
\]
Both denominators attain their minimum on $[0,1]$ at $x=1$, where they equal $1-d^2$. Each index is therefore strictly below $x\leq1$, so the cdf does not saturate and $T_D(x)<x$.
\end{proof}

\begin{theorem}\label{thm:creation}
Suppose
\begin{equation}
 \frac45<\kappa<1-d^2.
 \label{eq:openregion}
\end{equation}
Define
\[
 N_0(\kappa):=\min\{N\geq1:\kappa g_N>1\}.
\]
For every $N\geq N_0(\kappa)$, composition disclosure has an interior regular symmetric equilibrium, while secrecy and individual attribution have only pooling in the strategy family \eqref{eq:strategy}. The threshold is finite, and the explicit bound
\begin{equation}
 N\geq \left\lceil\frac{10\kappa}{5\kappa-4}\right\rceil
 \label{eq:sufficientN}
\end{equation}
is sufficient. The parameter region is nonempty whenever $d^2<1/5$.
\end{theorem}

\begin{proof}
For $N\geq N_0(\kappa)$, Theorem~\ref{thm:multiplier} gives $T'_{\Comp,N}(0)=\kappa g_N>1$, so $T_{\Comp,N}(x)-x>0$ just to the right of zero. At $x=1$, reports equal signals and orientations are ancillary for on-path reputation. Hence
\[
 T_{\Comp,N}(1)=\rho f\Delta_{\Comp,N}(1)
 =\frac{\kappa}{1-d^2}<1.
\]
Continuity gives an interior fixed point. Since $g_N\uparrow5/4$ and $\kappa>4/5$, the threshold $N_0(\kappa)$ is finite. Moreover,
\[
 g_N=\frac54-\frac5{2N}+\frac{N+5}{N2^{N+1}}
 >\frac54-\frac5{2N}.
\]
If the right-hand side of \eqref{eq:sufficientN} is strictly above $10\kappa/(5\kappa-4)$, the displayed lower bound gives $\kappa g_N>1$; if it equals that ratio, the strictly positive final term in $g_N$ gives the same conclusion. Proposition~\ref{prop:pooling} gives the remaining claims. The interval in \eqref{eq:openregion} is nonempty if and only if $d^2<1/5$.
\end{proof}

The theorem is an existence result and not a uniqueness claim, and it is worth being explicit about what survives alongside the informative equilibrium. Pooling remains a fixed point of the regular extension under $\Comp$, and it can also be supported with permissible off-path beliefs. Once $\kappa g_N>1$, however, pooling is locally unstable in the scalar regular best-response map, which is the sense in which composition disclosure \emph{creates} something. Section~\ref{sec:sparse} locates the positive branch, and Section~\ref{subsec:selection} establishes its robustness to a vanishing forced-fidelity perturbation without claiming unique selection. The sufficient panel size in \eqref{eq:sufficientN} diverges as $\kappa\downarrow4/5$, so the local creation mechanism is intrinsically a large-panel effect near the boundary.

\section{Information without attribution}\label{sec:information}

The observer has to act before the state is verified, and it is this observer---not the expert and not the market that scores her afterwards---whose information the disclosure policy is ultimately supposed to serve. Let
\[
 R_N=\sum_i\1\{r_i=+1\}
\]
be the number of positive reports. Conditional on $K=k$,
\begin{align}
 R_N\mid(\theta=+1,K=k)
 &\sim\Bin(k,1-x(1-Q))+\Bin(N-k,xQ),
 \label{eq:Rplus}\\
 R_N\mid(\theta=-1,K=k)
 &\sim\Bin(k,1-xQ)+\Bin(N-k,x(1-Q)),
 \label{eq:Rminus}
\end{align}
with independent summands. Thus $(K,R_N)$ is sufficient for the state under composition disclosure. The next result says something considerably stronger: nothing at all is lost.

\begin{theorem}\label{thm:blackwell}
For every finite $N$ and every fixed fidelity $x\in[0,1]$, the binary-state experiment generated by the named reports and $K$ under $\Comp$ is Blackwell-equivalent to the experiment generated by the named reports and the full orientation vector under $\Ind$. Equivalently, both Le Cam deficiencies are zero.
\end{theorem}

\begin{proof}
Erasing the identity-to-orientation mapping is a deterministic garbling from $\Ind$ to $\Comp$. For the reverse direction, define
\[
 D_N=R_N-K,\qquad
 \lambda=\frac{Q}{1-Q},\qquad
 a_x=\frac{1-xQ}{1-x(1-Q)}.
\]
For every individually attributed observation $(r,b)$ in the common support, direct multiplication of report probabilities gives
\begin{equation}
 \log\frac{\Prb(r,b\mid\theta=+1)}
 {\Prb(r,b\mid\theta=-1)}
 =D_N\log\lambda+(N-2K-D_N)\log a_x.
 \label{eq:exactLLR}
\end{equation}
The right-hand side depends only on $(K,R_N)$, observed under $\Comp$. The likelihood ratio is therefore constant across all orientation vectors compatible with a given composition observation. Conditional on that observation, the distribution of the missing orientation assignment is independent of the state. Drawing from this conditional distribution reconstructs an individually attributed observation through a state-independent Markov kernel. Hence $\Comp$ also garbles to $\Ind$. At $x=0$ the on-path assignment is $b=r$, while at $x=1$ orientations are ancillary; define the reconstruction kernel arbitrarily off the common support at these endpoints.
\end{proof}

\begin{corollary}\label{cor:orientationlaw}
Fix a common fidelity $x$. Theorem~\ref{thm:blackwell} remains true if the symmetric i.i.d.\ orientation prior is replaced by any commonly known joint law $\Pi_N$ on $\{-1,+1\}^N$ that is independent of the state, abilities, and signals. Thus orientations may be asymmetric, heterogeneous across names, or correlated across experts.
\end{corollary}

\begin{proof}
For every orientation vector in the support of $\Pi_N$,
\[
 \Prb(r,b\mid\theta)=\Pi_N(b)
 \prod_{i=1}^N P_Q^{b_i}(r_i\mid\theta,x).
\]
The factor $\Pi_N(b)$ cancels from the state likelihood ratio, and the remaining product is exactly \eqref{eq:exactLLR}, hence a function only of $(K,R_N)$. The conditional law of the missing assignment given $(r,K)$ is therefore state-independent, which supplies the same reverse Markov kernel as before.
\end{proof}

The theorem gives the paper its sharp design separation. Composition disclosure removes information about \emph{who} holds which conflict; it removes no information about \emph{which state} generated the panel. Any behavioural difference between $\Comp$ and $\Ind$ is therefore an incentive effect and not a mechanical loss of evidence available to the decision maker---which is exactly the counterfactual a sceptical regulator would want ruled out before agreeing to publish less.\footnote{The reconstruction is genuinely constructive: Appendix~\ref{app:likelihood} exhibits the Markov kernel that redistributes the missing conflicts across identities, and \eqref{eq:exactLLR} is the reason it does not depend on the state. In practice a decision maker never needs to perform the reconstruction; its existence is what licenses the claim that she has lost nothing by not being able to.}

Exact equivalence by itself does not establish that the count is the right boundary of orientation disclosure: perhaps an even coarser release would also preserve the state experiment. To answer that question, let an auditor use any committed, possibly randomised, state-independent disclosure kernel $\delta(\dd m\mid r,b)$ after reports are locked. Named reports remain public; the natural orientation-only policies are the special case $\delta(\dd m\mid b)$. Call the resulting policy \emph{state-lossless} if the experiment $(r,M)$ is Blackwell-equivalent to full attribution $(r,b)$.

\begin{theorem}\label{thm:minimaldisclosure}
Fix $x\in(0,1)$ and any commonly known orientation law $\Pi_N$ satisfying Corollary~\ref{cor:orientationlaw}. A disclosure kernel is state-lossless if and only if there is a measurable $k^*$ such that
\[
 k^*(r,M)=K(b)
\]
almost surely under the induced joint law in either state. For an orientation-only kernel, this is equivalent to recoverability from $M$ itself: for all supported $b,b'$ with $K(b)\neq K(b')$, the message laws $\delta(\cdot\mid b)$ and $\delta(\cdot\mid b')$ must be mutually singular. Hence $(r,K)$ is minimal in the refinement order among state-lossless public observations that retain the named report vector.
\end{theorem}

The proof is in Appendix~\ref{app:disclosure}. The key is to rewrite \eqref{eq:exactLLR} as
\begin{equation}
 \log\frac{\Prb(r,b\mid\theta=+1)}
 {\Prb(r,b\mid\theta=-1)}
 =R_N\log\lambda+(N-R_N)\log a_x
 -K\log(\lambda a_x).
 \label{eq:minimalLLR}
\end{equation}
For $Q>1/2$ and $x\in(0,1)$, $\lambda a_x>1$. Since every report vector has positive probability under every orientation vector, pooling two different counts changes the likelihood ratio at the same named report profile and therefore destroys information. The endpoints are deliberately excluded. At $x=0$, reports reproduce orientations and the state experiment is uninformative; at $x=1$, reports reproduce signals and orientations are ancillary. In either case the count need not be separately disclosed.

The qualification about named reports is essential. If reports could also be anonymised, $(R_N,K)$ would already be state-sufficient, and the likelihood-ratio statistic could be coarser still. Our minimality result concerns orientation disclosure because named reports are required for individual reputation. It is also an information-disclosure result, not an incentive-optimality claim: we compare three canonical regimes but do not establish that composition maximises candour among every state-lossless refinement of the count.

\subsection{The boundary of aggregate sufficiency}\label{subsec:alphabet}

The binary result is an instance of a general likelihood-separability condition. Let the state, report, and orientation sets $\Theta$, $\mathcal R$, and $\mathcal B$ be finite. Conditional on a state $\theta$ and orientation $b$, an expert's report has a common full-support kernel $p_\theta(r\mid b)$; reports are conditionally independent across experts. Let the joint orientation law $\Pi_N$ be state-independent. Write $\mathcal H(b)$ for the histogram of an orientation vector, fix a reference state $\theta_0$, and define
\[
 \ell_\theta(r,b)=\log\frac{p_\theta(r\mid b)}
 {p_{\theta_0}(r\mid b)}.
\]

\begin{theorem}\label{thm:alphabet}
Suppose $N\geq2$.
\begin{enumerate}[label=(\roman*),leftmargin=2.2em]
 \item If, for every $\theta\neq\theta_0$, there are functions $\alpha_\theta$ and $\beta_\theta$ such that
 \begin{equation}
   \ell_\theta(r,b)=\alpha_\theta(r)+\beta_\theta(b),
   \label{eq:additivellr}
 \end{equation}
 then named reports together with $\mathcal H(b)$ are Blackwell-equivalent to full attribution.
 \item If the orientation law has full support on $\mathcal B^N$, the converse holds: histogram disclosure is state-lossless only if \eqref{eq:additivellr} holds.
 \item Whenever \eqref{eq:additivellr} holds, $S_N$ below, restricted to $\operatorname{supp}\Pi_N$, is the minimal lossless orientation statistic:
 \begin{equation}
   S_N(b)=\left(\sum_{i=1}^N\beta_\theta(b_i)
   \right)_{\theta\neq\theta_0}.
   \label{eq:minimalscore}
 \end{equation}
 Any state-independent randomised garbling $(r,M)$ that is Blackwell-equivalent to full attribution must permit recovery of $S_N$ jointly from $(r,M)$; for an orientation-only policy, $M$ alone must reveal $S_N$.
\end{enumerate}
\end{theorem}

The condition is restrictive but economically interpretable. The additive decomposition is unique only up to state-specific constants, but those constants shift $S_N$ independently of $b$ and therefore leave its induced partition unchanged. For example,
\[
 p_\theta(r\mid b)=
 \frac{q_b(r)\exp\{\eta_\theta^{\mathsf T}T(r)\}}
 {Z_b(\eta_\theta)}
\]
satisfies it when $q_b(r)>0$: the state shifts a common report score, while orientations alter baseline weights and normalisations. By contrast, take two states, two reports, and orientations $A,B$, with
\[
 \begin{array}{c|cc}
 &p_0(1\mid b)&p_1(1\mid b)\\ \hline
 A&4/5&3/5\\
 B&3/10&1/2
 \end{array}.
\]
For reports $(1,0)$, the likelihood ratio is $15/28$ under assignment $(A,B)$ and $10/3$ under assignment $(B,A)$. The histogram is the same, so anonymous composition loses state information. The binary baseline succeeds because its log likelihood ratio has exactly the additive form, with the count $K$ as the minimal score.

\begin{remark}
Theorems~\ref{thm:blackwell} and~\ref{thm:minimaldisclosure} hold conditional on a common fixed fidelity $x$. When disclosure changes equilibrium behaviour---which is the whole point of Sections~\ref{sec:equilibrium} and~\ref{sec:sparse}---the induced experiments need not be equivalent, and Section~\ref{sec:sparse} shows that they are not. Corollary~\ref{cor:orientationlaw} removes symmetry, independence, and exchangeability from the orientation law for this information result. Theorem~\ref{thm:alphabet} identifies what replaces the binary structure: additive state likelihood ratios. Arbitrary behavioural heterogeneity need not have that form. The incentive theorems continue to use the symmetric i.i.d.\ binary benchmark.
\end{remark}

\section{Large panels and sparse accountability}\label{sec:sparse}

Theorem~\ref{thm:creation} establishes an interior equilibrium for large panels but does not say where it is. The relevant branch, close to the boundary $\kappa=4/5$, turns out to have the sparse scaling
\begin{equation}
 x_N=\frac{c}{N},\qquad c>0.
 \label{eq:sparsescaling}
\end{equation}
After the state is verified, normalise orientations and reports by
\[
 e_j=\1\{b_j=\theta\},\qquad y_j=\1\{r_j=\theta\}.
\]
Composition and verification reveal $L=\sum_je_j$, while the reports reveal $Y=\sum_jy_j$; their discrepancy is $Z_N=Y-L$. Under \eqref{eq:sparsescaling}, the opposed orientations that switch to correct evidence and the aligned orientations that switch to incorrect evidence converge to independent counts
\begin{equation}
 V\sim\Pois\!\left(\frac{Qc}{2}\right),
 \qquad
 U\sim\Pois\!\left(\frac{(1-Q)c}{2}\right),
 \qquad Z=V-U.
 \label{eq:poissonpair}
\end{equation}
Let $\varpi_z=\Prb(Z=z)$. Define
\begin{equation}
 A(c,Q)=\sum_{z\in\mathbb Z}\frac{\varpi_z^2}{\varpi_{z+1}},
 \qquad
 B(c,Q)=\sum_{z\in\mathbb Z}\frac{\varpi_z^2}{\varpi_{z-1}}.
 \label{eq:AB}
\end{equation}
Both sums are finite for $c>0$.

\begin{theorem}\label{thm:sparsemultiplier}
For every fixed $c>0$,
\begin{equation}
 \Delta_{\Comp,N}(c/N)=\frac{mc}{N}
 \left[3d+QA(c,Q)-(1-Q)B(c,Q)\right]+o(N^{-1}).
 \label{eq:sparsedelta}
\end{equation}
The convergence is uniform for $c$ in compact subsets of $(0,\infty)$. Relative to secrecy,
\begin{equation}
 \Gamma(c,Q):=\lim_{N\to\infty}
 \frac{\Delta_{\Comp,N}(c/N)}{\Delta_{\Hid}(c/N)}
 =\frac{3d+QA(c,Q)-(1-Q)B(c,Q)}{4d}.
 \label{eq:Gamma}
\end{equation}
Moreover, $\Gamma$ has a continuously differentiable extension to zero and
\begin{equation}
 \Gamma(c,Q)=\frac54-\frac{c}{16}+O(c^2)
\quad\text{as }c\downarrow0,
\qquad
 \lim_{c\to\infty}\Gamma(c,Q)=1.
\label{eq:Gammaexpansion}
\end{equation}
In particular, $\Gamma_c(c,Q)=-1/16+O(c)$ as $c\downarrow0$.
\end{theorem}

The proof is in Appendix~\ref{app:sparse}. The main technical difficulty is not pointwise Poisson convergence, which is routine, but convergence of the adjacent-mass ratios in \eqref{eq:AB}: these are ratios of probabilities of neighbouring discrepancies, and they are exactly what a reputational comparison of two adjacent reports asks for. The appendix supplies weighted uniform-integrability and posterior-expansion lemmas that control the finite-support boundaries. The universal coefficients in \eqref{eq:Gammaexpansion} show that the local $5/4$ overshoot is not an artefact of a particular signal quality.\footnote{As $c\downarrow0$ the composition advantage approaches $5/4$ and initially declines, because a single defection is then genuinely exceptional and the market's reassignment problem is highly sensitive to it. As $c\to\infty$ its marginal incentive advantage over secrecy converges to zero. The theorem does not assert global monotonicity of $\Gamma$.}

\begin{theorem}\label{thm:sparsebranch}
There exists $\eta>0$ such that, for every
$\kappa\in(4/5,4/5+\eta)$ with $\kappa<1-d^2$, the equation
\begin{equation}
 \kappa\Gamma(c^*,Q)=1
 \label{eq:cstar}
\end{equation}
has a unique solution in a neighborhood of zero. It satisfies
\begin{equation}
 c^*(\kappa)=25\left(\kappa-\frac45\right)
 +O\!\left((\kappa-\tfrac45)^2\right).
 \label{eq:cstarexpansion}
\end{equation}
There is a sequence of finite-panel regular composition equilibria $x_N^*$ with
\begin{equation}
 Nx_N^*\longrightarrow c^*(\kappa).
 \label{eq:xstar}
\end{equation}
The positive fixed point of the limiting rescaled best-response map is locally stable.
\end{theorem}

\begin{proof}
The continuously differentiable extension in Theorem~\ref{thm:sparsemultiplier} has $\Gamma(0,Q)=5/4$ and $\Gamma_c(0,Q)=-1/16$. Hence $\Gamma$ is strictly decreasing on a sufficiently small right neighborhood of zero. For $\kappa>4/5$ sufficiently close to $4/5$, continuity and monotonicity give a unique nearby positive solution to \eqref{eq:cstar}; substituting \eqref{eq:Gammaexpansion} gives \eqref{eq:cstarexpansion}. This is the one-sided version of the usual implicit-function argument.

For the finite-panel transfer, define
\[
 \Psi_N(c)=N T_{\Comp,N}(c/N)-c.
\]
On every compact subset of $(0,\infty)$ the uniform cdf is unsaturated for large $N$, and Theorem~\ref{thm:sparsemultiplier} gives
\begin{equation}
 \Psi_N(c)=c\{\kappa\Gamma(c,Q)-1\}+o(1)
 \label{eq:rescaledroot}
\end{equation}
uniformly. The limiting expression crosses zero strictly at $c^*$ because its derivative there is $\kappa c^*\Gamma_c(c^*,Q)<0$. Choose a compact bracket around $c^*$ on which the endpoint signs are strict. Uniform convergence and the intermediate-value theorem then give exact finite-$N$ roots in that bracket and imply \eqref{eq:xstar}. At the limiting fixed point, the derivative of $c\mapsto\kappa c\Gamma(c,Q)$ equals $1+\kappa c^*\Gamma_c(c^*,Q)$, which lies in $(0,1)$ for $\kappa$ sufficiently close to $4/5$.
\end{proof}

The theorem deliberately asserts stability only for the limiting rescaled map. A finite-$N$ stability claim would require $C^1$ convergence of the sparse incentive function rather than convergence of its level, and we have not established that. It is unnecessary for existence and for localising the branch, and we prefer to say so rather than to assert more than the proof delivers.

\subsection{State learning on the sparse branch}

Before verification, let $D_N=R_N-K$. Under the positive state and any sequence with $Nx_N\to c>0$,
\begin{equation}
 D_N\Rightarrow D=V-U,
 \label{eq:Dlimit}
\end{equation}
with $(V,U)$ as in \eqref{eq:poissonpair}; under the negative state the limit is $-D$. The Skellam identity
\[
 \frac{\varpi_z}{\varpi_{-z}}=
 \left(\frac{Q}{1-Q}\right)^z
\]
implies the limiting log-likelihood ratio
\begin{equation}
 \log\frac{\Prb(D=z\mid\theta=+1)}
 {\Prb(D=z\mid\theta=-1)}
 =z\log\frac{Q}{1-Q}.
 \label{eq:skellamLLR}
\end{equation}
The same limit follows from the exact likelihood ratio \eqref{eq:exactLLR}: since $\log a_{x_N}=-dx_N+O(x_N^2)$ and $N-2K-D_N=O_{\Prb}(\sqrt N)$, its second term is $o_{\Prb}(1)$. The disclosed composition thus functions as a benchmark against which the report profile is read, and the whole of the surviving evidence is carried by the signed distance between the two.

\begin{theorem}\label{thm:aggregation}
Under composition disclosure, let $x_N$ be any sequence such that $Nx_N\to c>0$. The rule that chooses $+1$ if $R_N>K$, chooses $-1$ if $R_N<K$, and randomises on equality is asymptotically optimal. Its limiting accuracy is
\begin{align}
 \mathcal A_{\Comp}(c)
 &=\Prb(V-U>0\mid\theta=+1)
 +\frac12\Prb(V-U=0\mid\theta=+1)\nonumber\\
 &=\frac12+\frac14\sum_{z\in\mathbb Z}
 |\varpi_z-\varpi_{-z}|>\frac12.
 \label{eq:accuracy}
\end{align}
Under hidden disclosure along the same sequence $x_N$, optimal accuracy converges to $1/2$.
\end{theorem}

\begin{proof}
Poisson convergence and \eqref{eq:skellamLLR} apply whenever $Nx_N\to c$. The exact likelihood ratio \eqref{eq:exactLLR} differs from $D_N\log\lambda$ by $o_{\Prb}(1)$, including on $D_N=0$; bounded convergence therefore shows that the sign rule attains the limiting Bayes risk. Since $Q>1/2$ and $c>0$, the reflected Skellam laws differ, so their total variation distance is positive. Under $\Hid$, each report agrees with the state with probability $1/2+dx_N/2$. The KL divergence between the two $N$-fold report laws is $O(Nx_N^2)=O(N^{-1})$; Pinsker's inequality therefore gives vanishing total variation.
\end{proof}

Combining Theorems~\ref{thm:blackwell}, \ref{thm:sparsebranch}, and \ref{thm:aggregation} gives the decision-quality logic of the paper in three steps. At a fixed fidelity, composition and attribution contain exactly the same state information. Endogenising fidelity can nevertheless make composition strictly more informative on the tremble-robust branch. The reason is not that anonymity adds evidence, but that anonymous composition produces the rare reports that attribution suppresses---and then, having disclosed the benchmark, makes them legible.\footnote{It is worth noting what this does \emph{not} say. Composition disclosure is not shown to dominate attribution at every parameter configuration; outside the region \eqref{eq:openregion} both regimes may support informative equilibria and the comparison becomes quantitative rather than qualitative. The claim is that there is an open set of environments on which the extreme regimes are silent and the intermediate one is not within the regular symmetric-cutoff family.}

\subsection{Tremble robustness}\label{subsec:selection}

The informative equilibrium under $\Comp$ coexists with pooling. We therefore ask what survives a small perturbation rather than impose a global equilibrium-selection claim. For $\eps>0$, suppose that, conditional on $s_i=-b_i$, every expert is forced to follow her signal with probability $\eps$ before making the cost-based choice. The perturbed best response is
\begin{equation}
 T_{D,N}^{\eps}(x)=\eps+(1-\eps)
 F\!\left(\rho\Delta_{D,N}(x)\right).
 \label{eq:perturbedBR}
\end{equation}

For fixed $N$, let $\mathcal A_{D,N}(x)$ be the observer's optimal state-classification accuracy under regime $D$ when experts use fidelity $x$.

\begin{theorem}\label{thm:trembleaccuracy}
Suppose costs satisfy \eqref{eq:uniform}, \eqref{eq:openregion} holds, and $N\geq N_0(\kappa)$. Then:
\begin{enumerate}[label=(\roman*),leftmargin=2.2em]
 \item every perturbed fixed point under $\Hid$ or $\Ind$ satisfies
 \begin{equation}
   x\leq\frac{\eps}{1-(1-\eps)\kappa/(1-d^2)}
   \longrightarrow0;
   \label{eq:extremetremblebound}
 \end{equation}
 \item there are $\delta_N>0$ and $\bar\eps_N>0$ such that, for every $\eps\in(0,\bar\eps_N)$, a perturbed composition equilibrium exists and every such equilibrium satisfies $x\geq\delta_N$;
 \item there is $\omega_N>0$ such that every perturbed composition equilibrium has accuracy at least $1/2+\omega_N$ for all sufficiently small $\eps$, whereas accuracy at every perturbed secrecy or attribution equilibrium converges uniformly to $1/2$ as $\eps\downarrow0$.
\end{enumerate}
\end{theorem}

\begin{proof}
For either extreme regime, Propositions~\ref{prop:pooling} and~\ref{prop:individual} imply
$T_{D,N}(x)\leq rx$, where $r=\kappa/(1-d^2)<1$. Any fixed point of \eqref{eq:perturbedBR} therefore satisfies \eqref{eq:extremetremblebound}.

For composition disclosure, $T'_{\Comp,N}(0)=\kappa g_N>1$. Hence there are $\delta_N,\eta_N>0$ such that
$T_{\Comp,N}(x)\geq(1+\eta_N)x$ on $(0,\delta_N]$. If $\eps<\eta_N/(1+\eta_N)$, then $T_{\Comp,N}^{\eps}(x)>x$ throughout $[0,\delta_N]$. Moreover,
\[
 T_{\Comp,N}^{\eps}(1)
 =\eps+(1-\eps)\frac{\kappa}{1-d^2}<1.
\]
Continuity gives a fixed point in $(\delta_N,1)$ and rules out every fixed point below $\delta_N$.

At $x=0$, the report experiment is state-independent under both extreme regimes; continuity and part~(i) give their accuracy limits. Under composition disclosure, the two state distributions differ at every $x>0$: the likelihood ratio in \eqref{eq:exactLLR} is nonconstant. Optimal binary-classification accuracy is continuous in $x$, so its minimum on the compact interval $[\delta_N,1]$ exceeds one half. This proves part~(iii).
\end{proof}

The result selects an accuracy comparison for each fixed qualifying $N$, not a unique composition equilibrium or a bound uniform in panel size. In the unperturbed game, pooling can still coexist.

\begin{corollary}\label{cor:sparsetremble}
Under the assumptions of Theorem~\ref{thm:sparsebranch}, let $\eps_N=o(1/N)$. There are perturbed composition equilibria $x_N^{\eps_N}$ such that $Nx_N^{\eps_N}\to c^*$.
\end{corollary}

\begin{proof}
For $c$ in a compact bracket around $c^*$,
\[
 N T_{\Comp,N}^{\eps_N}(c/N)-c
 =N\eps_N+(1-\eps_N)N T_{\Comp,N}(c/N)-c.
\]
By $N\eps_N\to0$ and \eqref{eq:rescaledroot}, the right-hand side converges uniformly to $c\{\kappa\Gamma(c,Q)-1\}$. The strict-crossing argument from Theorem~\ref{thm:sparsebranch} supplies perturbed roots in the bracket, and every such selected sequence converges to $c^*$ after multiplication by $N$.
\end{proof}

\section{Discussion}\label{sec:discussion}

\subsection{Institutional interpretation}

The model's central policy is a three-part protocol: commit publicly to a disclosure rule, collect sealed named reports, and then release the prescribed conflict statistic after filing. Three features of it matter, and each corresponds to an assumption that is doing real work.

First, the orientation must be directionally auditable. A formal constituency, a pre-registered model family, or a financial tie governed by an ex ante classification rule fits more naturally than an unverifiable private belief. The measurement and review rule must be credible and commonly known before reports are filed; the model takes that institutional input as given. The named classifications must also remain confidential until the committed release. If affiliations already reveal every expert's direction, an auditor cannot manufacture anonymity by withholding the mapping.

Second, orientation must be issue-specific. A pharmaceutical tie may orient an expert differently on approval, on pricing, and on safety. The policy should therefore publish the composition relevant to the decision at hand, and not a permanent moral label attached to the expert's name.

Third, reports must be locked before the composition is released. In a scientific advisory committee, members could submit signed probability assessments before the secretariat publishes the number of members whose classified ties point toward approval or rejection. In a forecasting organisation, analysts could file named directional forecasts before an auditor publishes the number whose institutional mandates or pre-registered model classes favour each outcome. Directionally auditable and eventually scorable forecasting tasks fit the model most literally; drug, environmental, and financial applications motivate the mechanism but need not have perfectly binary or quickly verified states. These examples are not assertions that current rules already implement composition-only disclosure.

The policy also has a privacy reading. An institution may hold a rich conflict register, while the model uses only its binary, issue-specific classifications. Publishing those classifications by name reveals a mapping with consequences outside the decision for which it was collected. Composition disclosure releases the state-sufficient count while protecting that binary mapping. In the baseline, this protection is not purchased with less state information, because Theorem~\ref{thm:blackwell} gives exact equivalence---a rare case in which a privacy constraint is free at the level of the decision problem. Its behavioural value comes from something else: it preserves responsibility while withholding individualised excuses.

Like aggregate releases in official statistics, the rule publishes a count rather than microdata. Unusually, it anonymises an \emph{inducement} rather than an outcome, and the anonymity strengthens individual incentives.

\subsection{What the mechanism is---and is not}

The baseline mechanism is not a Keynesian beauty contest, and it does not ask experts to match the crowd; experts care about posterior ability once the state has been verified, and about nothing else. Nor is the paper making the claim that unknown ideological priors generally make communication truthful. The verification-free extension below does create a form of peer validation, but only because the audited aggregate constraint turns the rest of the named report profile into a benchmark. Without that constraint, anonymous heterogeneity need not help at all, and may hurt.

Cross-identification is the source of the composition-versus-attribution wedge, but it is not by itself the whole baseline equilibrium. The open-region and sparse-branch results also use heterogeneity in ability, privately costly orientation defiance, eventual verification of the state, and a cost distribution with positive density near zero. Remove any one of these and the fixed point changes; remove more than one and the sparse branch may disappear entirely. Proposition~\ref{prop:noverification} isolates what remains locally when verification alone is removed.

The binary structure is doing useful work rather than merely simplifying. It makes composition one-dimensional and delivers the exact factorisation in \eqref{eq:exactLLR}. Theorem~\ref{thm:alphabet} supplies the boundary for multiple states, reports, or conflict categories: additive likelihood separability is sufficient for anonymous histograms to remain lossless and, under full-support orientation uncertainty, necessary. When it holds, \eqref{eq:minimalscore} identifies the minimal state-lossless orientation statistic; when it fails under full support, some assignment information is genuinely valuable for learning the state.

\subsection{Peer validation without ground truth}\label{subsec:noverification}

Suppose now that the state is never publicly revealed to the reputation market. Reputations are formed from the named reports and the committed orientation disclosure alone; the observer's action adds no information beyond that public record. Let $\widehat\Delta_{D,N}(x)$ denote the reputational return to following a conflicting signal in this variant, under the same regular extension at pooling.

\begin{proposition}\label{prop:noverification}
All three regular incentive levels at pooling equal zero. Their right derivatives satisfy
\[
 \widehat\Delta'_{\Hid,N}(0)
 =\widehat\Delta'_{\Ind,N}(0)=0.
\]
Under composition disclosure,
\begin{equation}
 \widehat\Delta'_{\Comp,N}(0)=md f_N=\frac{\chi}{4}f_N,
 \qquad
 f_N=3-\frac8N+\frac{N+4}{N2^{N-1}}.
 \label{eq:noverification}
\end{equation}
The sequence satisfies $f_1=0$, is strictly increasing, is positive for every $N\geq2$, and converges to $3$. Thus large-panel composition disclosure retains three quarters of the verified-state secrecy slope even though neither extreme regime has any first-order return.
\end{proposition}

The proof is in Appendix~\ref{app:disclosure}. The force is intrinsically collective. With no ground truth, one isolated orientation-defying report cannot establish ability. Once the count is known, however, a rare departure elsewhere in the panel changes whether the focal report is interpreted as evidence or mere orientation. These first-order peer events survive the averaging over the unknown state. Under individual attribution, history-by-history inference can remain, but its ex ante first-order value cancels; secrecy removes the accounting identity that links the reports.

Proposition~\ref{prop:noverification} is deliberately a local incentive theorem; it does not establish a full communication equilibrium without ground truth. The verified-state environment remains the basis of the equilibrium-creation, sparse-branch, and accuracy results.

\subsection{Scope and robustness}

Corollary~\ref{cor:orientationlaw} already settles one natural concern: exact state-information equivalence is not an artefact of symmetric or independent orientations. Any commonly known, state-independent joint orientation law is admissible at a common fixed fidelity. The symmetry and exchangeability assumptions enter instead through incentives, where they reduce the focal posterior to a scalar count and deliver the closed-form multiplier $g_N$.

The local incentive ranking is also robust to the cost distribution.

\begin{proposition}\label{prop:costrobustness}
Suppose $F(0)=0$ and $F$ is right-differentiable at zero with density $f_0\in(0,\infty)$. Let $\kappa_0=\rho f_0\chi$. Then
\[
 T'_{\Ind,N}(0)=\frac{\kappa_0}{4},\qquad
 T'_{\Hid,N}(0)=\kappa_0,\qquad
 T'_{\Comp,N}(0)=\kappa_0g_N.
\]
Consequently, whenever $g_N>1$ and $1/g_N<\kappa_0<1$, pooling is locally stable under both extreme regimes and locally unstable under composition disclosure within the regular symmetric-cutoff family.
\end{proposition}

\begin{proof}
Apply the right-hand chain rule to $T_{D,N}(x)=F(\rho\Delta_{D,N}(x))$ at zero and use Propositions~\ref{prop:hidden}--\ref{prop:individual} and Theorem~\ref{thm:multiplier}.
\end{proof}

Uniform costs are therefore used for the global claim that the extreme regimes have no other symmetric cutoff fixed point, not for the local $1:4:5$ mechanism. Beyond the baseline, asymmetric orientation laws generally change the incentive multiplier even though they leave the information theorem intact. Intrinsic concern for the observer's action and correlated signals would likewise change the behavioural fixed point. Multidimensional states or reports pose a separate informational question, for which Theorem~\ref{thm:alphabet} gives a sufficient likelihood-factorisation test that is also necessary under full-support orientation uncertainty.

\section{Conclusion}

Transparency does two distinct jobs in an expert panel. It helps an audience interpret what the panel has said, and it assigns explanations---or blame---to the individuals who said it. Full attribution performs both jobs at once, and because it performs them at once, the second job quietly undoes the first: the excuse it supplies is precisely what makes the evidence not worth filing.

In the model, the two jobs can be separated exactly. At any fixed interior reporting behaviour, composition disclosure retains every piece of state information available under individually attributed conflicts. Conditional on named reports remaining public, every lossless orientation disclosure must make the composition recoverable together with those reports. Beyond the binary model, additive likelihood separability is sufficient for anonymous histograms to remain lossless and, under full-support orientation uncertainty, necessary. Composition nevertheless changes reputational incentives, because the audience cannot attach the aggregate constraint to any particular expert. In large panels the residual cross-identification raises the local incentive slope above both secrecy and attribution, and creates an informative sparse branch on an open set where the two extreme regimes pool within the regular symmetric-cutoff family. For each fixed qualifying panel, a vanishing forced-fidelity perturbation sends every extreme-regime equilibrium to pooling while keeping every composition equilibrium bounded away from it.

Along the sparse branch, only order-one orientation-defying reports are filed---a handful in arbitrarily large panels---but the handful is enough, because the disclosed benchmark is what makes it readable. The same benchmark also produces positive local accountability when the state is never verified: rare departures validate one another through the audited count even though no external ground truth arrives.

The broader lesson is not that anonymity should replace transparency. It is that disclosure should be designed around the statistic that the social inference problem actually needs. When the assignment of that statistic serves mainly to supply private excuses, withholding the assignment can strengthen accountability without weakening aggregate learning---and an institution that publishes less about its members may end up knowing more about the world.

\appendix

\section{Proof of the finite-panel multiplier}\label{app:finite}

Normalise observations relative to the verified state:
\[
 e_j=\1\{b_j=\theta\},\qquad y_j=\1\{r_j=\theta\}.
\]
Let $n=N-1$ and let $J$ be the number of state-aligned peer orientations. Conditional on the state, $J\sim\Bin(n,1/2)$. Put $k=n-J$.

For a particular peer report vector containing $a$ correct reports, let $C_n(\ell,a;x)$ be the joint likelihood of that vector and exactly $\ell$ aligned peer orientations. Direct enumeration gives
\begin{align}
 2^nC_n(\ell,a;x)
 ={}&\sum_{u=\max\{0,\ell-(n-a)\}}^{\min\{\ell,a\}}
 \binom au\binom{n-a}{\ell-u}
 [1-x(1-Q)]^u(xQ)^{a-u}\nonumber\\
 &\qquad\times[x(1-Q)]^{\ell-u}
 [1-xQ]^{n-a-\ell+u}.
 \label{eq:Clocal}
\end{align}
Uniformly over feasible $a$,
\begin{align*}
 2^nC_n(a,a;x)
 &=1-x[a(1-Q)+(n-a)Q]+O(x^2),\\
 2^nC_n(a-1,a;x)
 &=aQx\{1-x[(a-1)(1-Q)+(n-a)Q]\}+O(x^3),\\
 2^nC_n(a+1,a;x)
 &=(n-a)(1-Q)x\{1-x[a(1-Q)+(n-a-1)Q]\}+O(x^3),\\
 2^nC_n(a-2,a;x)&=\binom a2Q^2x^2+O(x^3),\\
 2^nC_n(a+2,a;x)&=\binom{n-a}{2}(1-Q)^2x^2+O(x^3).
\end{align*}

Let $P_q^e(y)$ denote the focal report likelihood in normalised coordinates. For total aligned-orientation count $L$, define
\[
 \Lambda_q(y,L,a;x)=\frac12\left[
 P_q^0(y)C_n(L,a;x)+P_q^1(y)C_n(L-1,a;x)\right].
\]
This likelihood is affine in $q$, so $p\Lambda_h+(1-p)\Lambda_\ell=\Lambda_Q$. Moreover, $\partial_qP_q^e(y)=x(2y-1)$ for both orientations. Bayes' rule therefore yields the exact posterior-score identity
\begin{equation}
 \mu_y=p+mS_y,
 \qquad
 S_y(L,a;x)=
 \frac{x(2y-1)[C_n(L,a;x)+C_n(L-1,a;x)]}
 {P_Q^0(y)C_n(L,a;x)+P_Q^1(y)C_n(L-1,a;x)}.
 \label{eq:scoreidentity}
\end{equation}

Condition on $J=j$ and write $k=n-j$. For a focal orientation opposed to the state, put $R_o=S_1(j,a;x)-S_0(j,a;x)$; for an aligned focal orientation, put $R_a=S_0(j+1,a;x)-S_1(j+1,a;x)$. Substitution of the expansions above into \eqref{eq:scoreidentity} gives the complete one-departure calculation:
\begin{align*}
 a=j:\quad
 &R_o(0)=\frac1{Q(j+1)},
 &&R_o'(0)=1+\frac{j}{j+1},\\
 &R_a(0)=-\frac1{(1-Q)(k+1)},
 &&R_a'(0)=-1-\frac{k}{k+1};\\[.3em]
 a=j+1:\quad
 &R_o(0)=\frac2{Q(j+2)},
 &&R_a(0)=-\frac1{Q(j+2)};\\[.3em]
 a=j-1:\quad
 &R_o(0)=\frac1{(1-Q)(k+2)},
 &&R_a(0)=-\frac2{(1-Q)(k+2)}.
\end{align*}
No peer departs with probability $1-x\lambda_j+O(x^2)$, where $\lambda_j=j(1-Q)+kQ$. One opposed peer departs with probability $kQx+O(x^2)$ and one aligned peer departs with probability $j(1-Q)x+O(x^2)$.

Weighting the opposed focal case by $Q$ and the aligned case by $1-Q$ yields
\begin{align}
 \frac1m\left.\frac{\dd\Delta_{\Comp,N}(x\mid J=j)}{\dd x}\right|_{0+}
 ={}&Q\left(1+\frac{j}{j+1}\right)
 -(1-Q)\left(1+\frac{k}{k+1}\right)\nonumber\\
 &+\lambda_j\left(\frac1{k+1}-\frac1{j+1}\right)
 +\frac{k(3Q-1)}{j+2}
 +\frac{j(3Q-2)}{k+2}.
 \label{eq:conditionalderivative}
\end{align}
At level zero the conditional contribution equals $m[(j+1)^{-1}-(k+1)^{-1}]$, whose expectation vanishes by binomial symmetry. Hence $\Delta_{\Comp,N}(0)=0$ under the regular extension.

Taking expectations in \eqref{eq:conditionalderivative}, using $J\overset d=n-J$, and collecting terms gives
\begin{equation}
 \frac{\Delta'_{\Comp,N}(0)}m
 =d\E\!\left[
 1+\frac{2J}{J+1}+\frac{3J}{n-J+2}-\frac{J}{n-J+1}
 \right].
 \label{eq:derivativereduction}
\end{equation}
For $J\sim\Bin(N-1,1/2)$, binomial integration identities imply
\begin{align}
 \E\frac{J}{J+1}&=1-\frac2N+\frac{2^{1-N}}N,
 \label{eq:binid1}\\
 \E\frac{J}{N-J+1}&=1-\frac2N+\frac{2^{1-N}}N,
 \label{eq:binid2}\\
 \E\frac{J}{N-J}&=1-2^{1-N}.
 \label{eq:binid3}
\end{align}
The first identity follows by writing $J/(J+1)=1-(J+1)^{-1}$ and integrating $(1+t)^{N-1}$ on $[0,1]$. For the second, symmetry gives $\E[J/(N-J+1)]=\E[(n-J)/(J+2)]$, and a binomial index shift makes the latter equal to $\E[J/(J+1)]$. The third follows by the same change of index.

Since $\Delta'_{\Hid}(0)=4md$, equations \eqref{eq:derivativereduction}--\eqref{eq:binid3} give
\[
 g_N=\frac14\left(1+2\E\frac{J}{J+1}
 +3\E\frac{J}{N-J+1}-\E\frac{J}{N-J}\right)
 =\frac54-\frac{5}{2N}+\frac{N+5}{N2^{N+1}}.
\]
Finally,
\[
 g_{N+1}-g_N=
 \frac{5\,2^{N+1}-N^2-6N-10}{N(N+1)2^{N+2}}>0.
\]
The numerator is positive at $N=1$, and its forward difference is $5\,2^{N+1}-2N-7>0$. Direct substitution gives $g_9<1<g_{10}$, while $g_1=1/4$ and the limit is $5/4$. This proves Theorem~\ref{thm:multiplier}.

\section{Proofs for disclosure and peer validation}\label{app:disclosure}

\subsection{Minimality and finite alphabets}

\begin{proof}[Proof of Theorem~\ref{thm:minimaldisclosure}]
For $x\in(0,1)$, every report vector has positive conditional probability under every orientation vector in the maintained support. Moreover,
\[
 Q(1-xQ)-(1-Q)\{1-x(1-Q)\}=(2Q-1)(1-x)>0,
\]
so $\lambda a_x>1$. At any fixed $r$, the likelihood ratio in \eqref{eq:minimalLLR} is consequently injective in $K$.

Let $X=(r,b)$, $Y=(r,M)$, and let $L_X=\dd P_+^X/\dd P_-^X$. Because $Y$ is a garbling of $X$,
\[
 L_Y(Y)=\E_-[L_X(X)\mid Y].
\]
The data-processing identity for chi-squared divergence is
\[
 \chi^2(P_+^X\Vert P_-^X)-\chi^2(P_+^Y\Vert P_-^Y)
 =\E_-\!\left[\operatorname{Var}_-(L_X\mid Y)\right].
\]
Blackwell equivalence forces equality of the divergences, hence $L_X$ is measurable with respect to $Y$. At each fixed $r$, equation~\eqref{eq:minimalLLR} is injective in $K$, so there must be a measurable $k^*$ with $k^*(r,M)=K(b)$ almost surely under the induced joint law in either state. For an orientation-only kernel, every report vector lies in every supported orientation vector's full-support report law. Thus if message laws associated with two different compositions were not mutually singular, two different full-data likelihood ratios would occur over the same $(r,M)$ on a set of positive measure. Consequently $M$ itself must reveal $K$.

Conversely, if $(r,M)$ reveals $K$, the message experiment refines composition disclosure. Theorem~\ref{thm:blackwell} makes composition disclosure equivalent to full attribution, while $(r,M)$ remains a garbling of full attribution. The experiments are therefore equivalent.
\end{proof}

\begin{proof}[Proof of Theorem~\ref{thm:alphabet}]
The state-independent factor $\Pi_N(b)$ cancels from every full-data likelihood ratio. Under \eqref{eq:additivellr},
\begin{equation}
 \log\frac{\Prb_\theta(r,b)}{\Prb_{\theta_0}(r,b)}
 =\sum_i\alpha_\theta(r_i)+\sum_i\beta_\theta(b_i).
 \label{eq:aggregateadditive}
\end{equation}
The first sum is observed from the named reports and the second is a function of the orientation histogram. Every state likelihood ratio is therefore measurable with respect to $(r,\mathcal H(b))$. Conditional on these data, the missing assignment of orientations is state-independent and supplies the reverse Blackwell kernel.

For necessity, apply the conditional-variance identity from the preceding proof separately to each pair $(P_\theta,P_{\theta_0})$. Blackwell equivalence makes every likelihood ratio $\dd P_\theta/\dd P_{\theta_0}$ measurable with respect to the histogram experiment. Now fix $r,r'\in\mathcal R$ and $b,b'\in\mathcal B$, holding all other coordinates fixed when $N>2$. Swapping $b$ and $b'$ across the two named reports preserves the histogram. Full support and likelihood-ratio measurability give, for every $\theta\neq\theta_0$,
\[
 \ell_\theta(r,b)+\ell_\theta(r',b')
 =\ell_\theta(r,b')+\ell_\theta(r',b).
\]
Thus every two-by-two cross-difference of $\ell_\theta$ is zero. Fixing $r_0,b_0$ and defining
\[
 \alpha_\theta(r)=\ell_\theta(r,b_0),
 \qquad
 \beta_\theta(b)=\ell_\theta(r_0,b)-\ell_\theta(r_0,b_0)
\]
gives \eqref{eq:additivellr}.

Finally, \eqref{eq:aggregateadditive} depends on orientations exactly through the score vector $S_N$. Revealing that vector suffices. Conversely, the conditional-variance argument from the preceding proof makes every state likelihood ratio measurable with respect to $(r,M)$. At fixed $r$, distinct supported score vectors give distinct likelihood-ratio vectors, so $(r,M)$ must reveal $S_N$ on $\operatorname{supp}\Pi_N$. If the message is orientation-only, full support of reports further implies mutually singular message laws across distinct supported scores, and $M$ alone reveals the score.
\end{proof}

\subsection{No verification}

For clarity, put hats on posteriors formed without observing the state. Averaging \eqref{eq:focalL} over the state prior gives
\[
 \widehat\Lambda_t(r,K,a;x)=
 \frac12\sum_{\theta\in\{-1,+1\}}
 \Lambda_t(r,K,a\mid\theta,x),
 \qquad
 \widehat\mu(r,K,a;x)=
 \frac{p\widehat\Lambda_H}
 {p\widehat\Lambda_H+(1-p)\widehat\Lambda_L}.
\]
For the normalization $b_i=+1,s_i=-1$, the exact no-verification incentive is
\begin{align}
 \widehat\Delta_{\Comp,N}(x)
 ={}&\sum_{\theta}\Prb(\theta\mid s_i=-1)
 \sum_{k,a}P_{N-1,\theta}(k,a;x)\nonumber\\
 &\quad\times\left[
 \widehat\mu(-1,k+1,a;x)-
 \widehat\mu(+1,k+1,a;x)\right].
 \label{eq:noverificationexact}
\end{align}

\begin{proof}[Proof of Proposition~\ref{prop:noverification}]
Let $n=N-1$, and continue to use the state-normalised coordinates $(y,L,a)$ from Appendix~\ref{app:finite}. State averaging pairs a history with its reflection. Define
\[
 \overline\Lambda_q(y,L,a;x)=\frac12\left\{
 \Lambda_q(y,L,a;x)+\Lambda_q(1-y,N-L,n-a;x)
 \right\}.
\]
This likelihood is affine in $q$, so $\widehat\mu_y=p+m\widehat S_y$. Suppressing $n$ and $x$ in $C_n$ and taking $C_n(\ell,a;x)=0$ outside its feasible support, direct differentiation in $q$ gives
\begin{align}
 \widehat S_y(L,a;x)
 ={}&\frac{x(2y-1)}{D_y(L,a;x)}
 \bigl\{C(L,a)+C(L-1,a)\nonumber\\
 &\hspace{6em}-C(N-L,n-a)-C(N-L-1,n-a)\bigr\},
 \label{eq:noverifyscore}
\end{align}
where
\begin{align*}
 D_y(L,a;x)={}&P_Q^0(y)C(L,a)+P_Q^1(y)C(L-1,a)\\
 &+P_Q^0(1-y)C(N-L,n-a)
 +P_Q^1(1-y)C(N-L-1,n-a).
\end{align*}
This identity makes the latent-state pairing explicit rather than hiding it inside \eqref{eq:noverificationexact}.

Let $J\sim\Bin(n,1/2)$ be the number of peer orientations aligned with the latent state and write $k=n-j$. Conditional on $J=j$, put
\begin{align*}
 B_j(a;x)={}&Q\{\widehat S_1(j,a;x)-\widehat S_0(j,a;x)\}\\
 &+(1-Q)\{\widehat S_0(j+1,a;x)-\widehat S_1(j+1,a;x)\}.
\end{align*}
Let $s=Q^2+(1-Q)^2$. Substitution of the five local expansions in \eqref{eq:Clocal} into \eqref{eq:noverifyscore} gives the complete zero- and one-departure table:
\begin{align*}
 B_j(j;0)&=0,\\
 \partial_xB_j(j;0)
 &=d\left[
 Q\frac{2j-k}{j+1}+(1-Q)\frac{2k-j}{k+1}
 \right],\\
 B_j(j+1;0)&=\frac{2Qd}{s(j+2)},
 &B_j(j-1;0)&=\frac{2(1-Q)d}{s(k+2)}.
\end{align*}
The $j+1$ entry is used only when $k>0$, and the $j-1$ entry only when $j>0$. Conditional peer probabilities are, respectively,
\[
 1-x\{j(1-Q)+kQ\}+O(x^2),\qquad
 kQx+O(x^2),\qquad j(1-Q)x+O(x^2).
\]
Histories with two or more departures are $O(x^2)$. Since the zero-departure level in the table is zero, differentiation yields
\begin{align}
 \frac{\widehat\Delta'_{\Comp,N}(0\mid J=j)}{md}
 ={}&Q\frac{2j-k}{j+1}+(1-Q)\frac{2k-j}{k+1}\nonumber\\
 &+\frac{2kQ^2}{s(j+2)}+
 \frac{2j(1-Q)^2}{s(k+2)}.
 \label{eq:noverificationconditional}
\end{align}
Average \eqref{eq:noverificationconditional} and use $J\overset d=n-J$. The terms containing $Q$ combine, leaving
\[
 \frac{\widehat\Delta'_{\Comp,N}(0)}{md}
 =\E\left[
 \frac{2J-(n-J)}{J+1}+\frac{2(n-J)}{J+2}
 \right].
\]
The binomial index shift
$\E[(n-J)/(J+2)]=\E[J/(J+1)]$, together with symmetry
$\E[(n-J)/(J+1)]=\E[J/(N-J)]$, gives
\begin{equation}
 \frac{\widehat\Delta'_{\Comp,N}(0)}{md}
 =4\E\!\left[\frac{J}{J+1}\right]
 -\E\!\left[\frac{J}{N-J}\right].
 \label{eq:noverificationreduction}
\end{equation}
The two expectations are evaluated in \eqref{eq:binid1} and \eqref{eq:binid3}. Substitution yields
\[
 4\left(1-\frac2N+\frac{2^{1-N}}N\right)
 -\left(1-2^{1-N}\right)
 =3-\frac8N+\frac{N+4}{N2^{N-1}}=f_N,
\]
which proves \eqref{eq:noverification}.

It remains to verify the two zero slopes. Under secrecy, let $L_\theta(z;x)$ be the likelihood of a peer report profile $z$ after orientations are integrated out. The focal public likelihood without verification is
\[
 \overline\Lambda_q^{\Hid}(r,z;x)=\frac14\left\{
 [1+x(2q-1)r]L_{+}(z;x)
 +[1-x(2q-1)r]L_{-}(z;x)
 \right\}.
\]
At pooling, $L_+-L_-=O(x)$ uniformly over the finite set of $z$, so the component depending on $q$ is $O(x^2)$. Both intervened expected reputations are consequently $p+O(x^2)$.

Under attribution, condition on the disclosed peer orientations and again write $L_\theta(z;x)$ for the peer-data likelihood. A focal report $-b$ reveals the conflicting signal $s=-b$; because the reporting tremble and $v_i$ are independent of type and the marginal signal is type-independent, iterated expectations make its expected posterior exactly $p$. For a focal report $b$,
\[
 \overline\Lambda_q^{\Ind}(b,z;x)=\frac12\left\{
 [1-x(1-q)]L_b(z;x)+[1-xq]L_{-b}(z;x)
 \right\},
\]
and therefore $\partial_q\overline\Lambda_q^{\Ind}=(x/2)(L_b-L_{-b})$. On zero-departure peer histories the difference is $O(x)$, while one-departure histories have probability $O(x)$. The expected posterior after pandering is thus $p+O(x^2)$. Hence
$\widehat\Delta'_{\Hid,N}(0)=\widehat\Delta'_{\Ind,N}(0)=0$.
The displayed likelihoods and score table also show $\widehat\Delta_{D,N}(0)=0$ for every regime under the regular extension.

Finally,
\[
 f_{N+1}-f_N=
 \frac{8\,2^N-(N^2+5N+8)}{N(N+1)2^N}>0.
\]
The numerator is positive at $N=1$ and its forward difference is $8\,2^N-2N-6>0$ thereafter. The endpoint values, positivity, and limit follow directly.
\end{proof}

\section{Proof of the sparse multiplier}\label{app:sparse}

We first establish the adjacent-ratio convergence used to pass from finite binomial discrepancies to the Skellam experiment.

Let $n=N-1$. When exactly $\ell$ peers are aligned with the verified state, define
\begin{equation}
 P_{\ell,N}^{c}(z)=\Prb(V_{\ell,N}-U_{\ell,N}=z),
 \quad
 V_{\ell,N}\sim\Bin(n-\ell,cQ/N),
 \quad
 U_{\ell,N}\sim\Bin(\ell,c(1-Q)/N),
 \label{eq:peerdiscrepancy}
\end{equation}
with independent counts.

\begin{lemma}\label{lem:adjacent}
Fix a compact interval $C\subset(0,\infty)$ and let
\[
 \mathcal L_N=\left\{\ell\in\{1,\ldots,n-1\}:
 \left|\ell-\frac n2\right|\leq N^{2/3}\right\}.
\]
Uniformly over $c\in C$ and $\ell\in\mathcal L_N$,
\begin{align*}
 \sum_z\frac{[P_{\ell,N}^{c}(z)]^2}
 {P_{\ell-1,N}^{c}(z+1)}&\longrightarrow A(c,Q),\\
 \sum_z\frac{[P_{\ell,N}^{c}(z)]^2}
 {P_{\ell+1,N}^{c}(z-1)}&\longrightarrow B(c,Q),
\end{align*}
where each sum is over the support of its numerator. The summands are uniformly integrable.
\end{lemma}

\begin{proof}
We prove the first statement; the second is symmetric. The denominator has no boundary zero because the support of $P_{\ell-1,N}^{c}(z+1)$ contains the support of $P_{\ell,N}^{c}(z)$. Let $f_z$ be the discrepancy law with $n-\ell$ opposed trials and $\ell-1$ aligned trials. With $a=Qc/N$ and $b=(1-Q)c/N$,
\[
 P_{\ell,N}^{c}(z)=(1-b)f_z+bf_{z+1},
 \qquad
 P_{\ell-1,N}^{c}(z+1)=af_z+(1-a)f_{z+1}.
\]
For nonnegative $r,s$, $(Ar+Bs)/(Cr+Ds)\leq\max\{A/C,B/D\}$. Hence
\begin{equation}
 \frac{P_{\ell,N}^{c}(z)}{P_{\ell-1,N}^{c}(z+1)}\leq C_0N
 \label{eq:globalratio}
\end{equation}
uniformly on the stated parameter set and support.

Put $J_N=4\log N/\log\log N$. For $\ell\in\mathcal L_N$, the two rare-count means equal $Qc/2+O(N^{-1/3})$ and $(1-Q)c/2+O(N^{-1/3})$. Uniformly for count realisations no larger than $2J_N$, binomial masses equal their Poisson counterparts times
\[
 1+O\!\left(\frac{(J_N+1)^2}{N}+N^{-1/3}(J_N+1)\right)=1+o(1).
\]
To justify this approximation relative to probabilities on the moving window, let $a_0>0$ be a uniform lower bound for the two limiting Poisson means and let $a_1$ bound their sum from above; such constants exist because $C$ is compact in $(0,\infty)$. For $|z|\leq J_N$, the event that the count with the required sign equals $|z|$ and the other count is zero gives
\[
 \varpi_z\geq e^{-a_1}\frac{a_0^{|z|}}{|z|!}
 \geq N^{-4+o(1)}.
\]
Meanwhile, for the sum $T_N$ of the two finite-sample rare counts, a Chernoff bound gives
\[
 \Prb(T_N>2J_N)
 \leq\left(\frac{eC_1}{2J_N}\right)^{2J_N}
 =N^{-8+o(1)}
\]
for a uniform $C_1$. Thus configurations omitted by truncating the joint counts at $2J_N$ are $o(\varpi_z)$ uniformly on $|z|\leq J_N$. Combined with the preceding factorial expansion, this yields, uniformly on that window,
\[
 \frac{P_{\ell,N}^{c}(z)}{\varpi_z}=1+o(1),
 \qquad
 \frac{P_{\ell-1,N}^{c}(z+1)}{\varpi_{z+1}}=1+o(1).
\]
The truncated sum therefore converges to $\sum_{|z|\leq J_N}\varpi_z^2/\varpi_{z+1}$.

For the finite-sample tail, \eqref{eq:globalratio} and a Chernoff bound give
\[
 \sum_{|z|>J_N}\frac{[P_{\ell,N}^{c}(z)]^2}
 {P_{\ell-1,N}^{c}(z+1)}
 \leq C_0N\Prb(V_{\ell,N}+U_{\ell,N}>J_N)=o(1).
\]
For the limiting tail, set $\alpha=Qc/2$. The Poisson identity
$\E[V\mid Z=z]=\alpha\varpi_{z-1}/\varpi_z$ implies, after relabeling,
\[
 \sum_{|z|>M}\frac{\varpi_z^2}{\varpi_{z+1}}
 \leq\frac1{\alpha^2}\E\!\left[V^2\1\{V+U>M-1\}\right],
\]
which vanishes uniformly on $C$. This proves convergence and uniform integrability. The proof for $B$ uses $\E[U\mid Z=z]=\beta\varpi_{z+1}/\varpi_z$, where $\beta=(1-Q)c/2$.
\end{proof}

Let $R_y^N(z,L)$ be the focal posterior of high ability when the focal correct-report indicator is $y\in\{0,1\}$, total discrepancy is $z$, and the disclosed number of state-aligned orientations is $L$. Conditional on $L$, the market assigns the focal orientation weights $\alpha_L=(N-L)/N$ and $\beta_L=L/N$. Exact focal likelihoods are
\begin{align}
 \Lambda_t(1,z,L)
 &=\alpha_L\frac cNq_tP_{L,N}^{c}(z-1)
 +\beta_L\left[1-\frac cN(1-q_t)\right]P_{L-1,N}^{c}(z),
 \label{eq:sparseL1}\\
 \Lambda_t(0,z,L)
 &=\alpha_L\left(1-\frac cNq_t\right)P_{L,N}^{c}(z)
 +\beta_L\frac cN(1-q_t)P_{L-1,N}^{c}(z+1).
 \label{eq:sparseL0}
\end{align}
Because the likelihoods are affine in $q_t$, Bayes' rule gives
\begin{align}
 R_1^N(z,L)-p
 &=\frac{m(c/N)[\alpha_LP_{L,N}^{c}(z-1)+\beta_LP_{L-1,N}^{c}(z)]}
 {\alpha_L(c/N)QP_{L,N}^{c}(z-1)+
 \beta_L[1-(c/N)(1-Q)]P_{L-1,N}^{c}(z)},
 \label{eq:exactR1}\\
 R_0^N(z,L)-p
 &=-\frac{m(c/N)[\alpha_LP_{L,N}^{c}(z)+\beta_LP_{L-1,N}^{c}(z+1)]}
 {\alpha_L[1-(c/N)Q]P_{L,N}^{c}(z)+
 \beta_L(c/N)(1-Q)P_{L-1,N}^{c}(z+1)}.
 \label{eq:exactR0}
\end{align}
These mixture denominators also handle the finite-support boundary cases.

For $|L-N/2|\leq N^{2/3}$, binomial-to-Poisson approximation and first-order expansion give, uniformly on each fixed finite set of $z$ and compact $c$-sets,
\begin{align}
 R_1^N(z,L)&=p+\frac{mc}{N}
 \left(1+\frac{\varpi_{z-1}}{\varpi_z}\right)+o(N^{-1}),
 \label{eq:R1exp}\\
 R_0^N(z,L)&=p-\frac{mc}{N}
 \left(1+\frac{\varpi_{z+1}}{\varpi_z}\right)+o(N^{-1}).
 \label{eq:R0exp}
\end{align}
\begin{lemma}\label{lem:weightedposteriors}
Fix a compact interval $C\subset(0,\infty)$. Uniformly over $c\in C$ and integer $L$ satisfying $|L-N/2|\leq N^{2/3}$,
\begin{align}
 \frac{N}{mc}\sum_z P_{L,N}^{c}(z)
 \bigl[R_1^N(z+1,L)-R_0^N(z,L)\bigr]
 &\longrightarrow 3+A(c,Q),
 \label{eq:weightedRopposed}\\
 \frac{N}{mc}\sum_z P_{L-1,N}^{c}(z)
 \bigl[R_0^N(z-1,L)-R_1^N(z,L)\bigr]
 &\longrightarrow-3-B(c,Q).
 \label{eq:weightedRaligned}
\end{align}
\end{lemma}

\begin{proof}
We give the argument for \eqref{eq:weightedRopposed}; the other statement is symmetric. Put $P_z=P_{L,N}^{c}(z)$, $\widetilde P_z=P_{L-1,N}^{c}(z+1)$, and retain $J_N=4\log N/\log\log N$ from Lemma~\ref{lem:adjacent}. Substitution in \eqref{eq:exactR1}--\eqref{eq:exactR0}, together with $\alpha_L/\beta_L=1+O(N^{-1/3})$, gives uniformly for $|z|\leq J_N$,
\begin{align*}
 \frac{N}{mc}\bigl[R_1^N(z+1,L)-p\bigr]
 &=1+\frac{P_z}{\widetilde P_z}+o(1),\\
 \frac{N}{mc}\bigl[p-R_0^N(z,L)\bigr]
 &=1+\frac{\widetilde P_z}{P_z}+o(1).
\end{align*}
The $o(1)$ terms are uniform: on this window the binomial masses and their adjacent ratios equal their Skellam counterparts up to $1+o(1)$, exactly as in the central-window part of Lemma~\ref{lem:adjacent}. Multiplying by $P_z$ and summing gives $1+A(c,Q)$ from the first line and $2$ from the second. For the complement, posterior differences are bounded by one, while the Chernoff estimate used in Lemma~\ref{lem:adjacent} gives
\[
 N\Prb\{|V_{L,N}-U_{L,N}|>J_N\}
 \leq N\Prb\{V_{L,N}+U_{L,N}>J_N\}=o(1)
\]
uniformly on $C$. The limiting adjacent-ratio tail is uniformly negligible by Lemma~\ref{lem:adjacent}. This proves \eqref{eq:weightedRopposed}; replacing $L$ by $L-1$ and reversing the adjacent shift proves \eqref{eq:weightedRaligned}.
\end{proof}

Hoeffding's inequality makes compositions outside the $N^{2/3}$ window exponentially negligible. Lemma~\ref{lem:weightedposteriors} therefore applies to the unconditional reputational comparison.

If the focal orientation is opposed to the state, reporting the signal changes its correct-report indicator from zero to one and raises total discrepancy from the peer value $Z$ to $Z+1$. Lemma~\ref{lem:weightedposteriors} yields
\[
 \E[R_1^N(Z+1,L)-R_0^N(Z,L)]
 =\frac{mc}{N}[3+A(c,Q)]+o(N^{-1}).
\]
If the focal orientation is aligned, following a conflicting signal changes its indicator from one to zero and lowers discrepancy from $Z$ to $Z-1$, so
\[
 \E[R_0^N(Z-1,L)-R_1^N(Z,L)]
 =-\frac{mc}{N}[3+B(c,Q)]+o(N^{-1}).
\]
The state agrees with a conflicting signal with probability $Q$ and with the orientation with probability $1-Q$. Weighting these expressions proves \eqref{eq:sparsedelta}; dividing by
$\Delta_{\Hid}(c/N)=4mdc/N+o(N^{-1})$ proves \eqref{eq:Gamma}.

It remains to establish \eqref{eq:Gammaexpansion} with a differentiable remainder. Put $t=c/2$, $q=Q$, and $r=1-Q$, so $d=q-r$. Write $\varpi_z=e^{-t}s_z(t)$. For every integer $k\geq0$,
\[
 s_k(t)=\sum_{u\geq0}\frac{(qt)^{u+k}(rt)^u}{(u+k)!u!},
 \qquad
 s_{-k}(t)=\sum_{u\geq0}\frac{(rt)^{u+k}(qt)^u}{(u+k)!u!}.
\]
Thus
\[
 A(c,Q)=e^{-t}\sum_{z\in\mathbb Z}
 \frac{s_z(t)^2}{s_{z+1}(t)}.
\]
Only the terms $z=0,1,2$ contribute through order $t$:
\[
\begin{array}{c@{\qquad}l}
 z & e^{-t}s_z(t)^2/s_{z+1}(t)\\ \hline
 0 & \displaystyle \frac1{qt}-\frac1q+
 \left(\frac1{2q}+\frac{3r}{2}\right)t+O(t^2)\\[.6em]
 1 & 2-2t+O(t^2)\\
 2 & \displaystyle \frac{3q}{2}t+O(t^2).
\end{array}
\]
To control the omitted terms and their derivatives, factor
\[
 s_k(t)=\frac{(qt)^k}{k!}u_k(t),\qquad
 s_{-k}(t)=\frac{(rt)^k}{k!}u_k(t),\qquad
 u_k(t)=\sum_{u\geq0}
 \frac{(qrt^2)^u k!}{u!(k+u)!}.
\]
On every interval $0\leq t\leq t_0$, the functions $u_k$, $1/u_k$, and their first derivatives are bounded uniformly in $k$: indeed,
\[
 1\leq u_k(t)\leq e^{qrt^2},
 \qquad |u_k'(t)|\leq2qrt e^{qrt^2}.
\]
Moreover,
\begin{align*}
 \frac{s_k(t)^2}{s_{k+1}(t)}
 &=\frac{k+1}{k!}(qt)^{k-1}
 \frac{u_k(t)^2}{u_{k+1}(t)}, &&k\geq3,\\
 \frac{s_{-k}(t)^2}{s_{-(k-1)}(t)}
 &=\frac{(rt)^{k+1}}{k^2(k-1)!}
 \frac{u_k(t)^2}{u_{k-1}(t)}, &&k\geq1.
\end{align*}
Summable factorial majorants, also after termwise differentiation, therefore imply
\[
 \sum_{k\geq3}\frac{s_k(t)^2}{s_{k+1}(t)}=O(t^2),
 \qquad
 \sum_{k\geq1}\frac{s_{-k}(t)^2}{s_{-(k-1)}(t)}=O(t^2),
\]
with derivatives $O(t)$. Summing the table gives, with a remainder satisfying $R_A(t)=O(t^2)$ and $R_A'(t)=O(t)$,
\[
 A(c,Q)=\frac1{qt}+2-\frac1q+\frac{r}{2q}t+R_A(t)
 =\frac{2}{Qc}+2-\frac1Q+\frac{1-Q}{4Q}c+O(c^2).
\]
Reflection of the Skellam law interchanges $q$ and $r$, and hence
\[
 B(c,Q)=\frac{2}{(1-Q)c}+2-\frac1{1-Q}
 +\frac{Q}{4(1-Q)}c+O(c^2),
\]
again with derivative remainder $O(c)$. Consequently,
\[
 qA-rB=2(q-r)-\frac{q-r}{2}t+O(t^2),
\]
where the derivative of the remainder is $O(t)$. Since $t=c/2$ and $d=q-r$,
\[
 \Gamma(c,Q)=\frac{3d+qA-rB}{4d}
 =\frac54-\frac{c}{16}+O(c^2),
 \qquad
 \Gamma_c(c,Q)=-\frac1{16}+O(c).
\]
This proves both the continuously differentiable extension at zero and the stated expansion.

For the high-intensity limit, let $\alpha=Qc/2$ and $\beta=(1-Q)c/2$. Poisson size biasing gives
\[
 A(c,Q)=1+\frac{\operatorname{Var}(\E[V\mid Z])}{\alpha^2},
 \qquad
 B(c,Q)=1+\frac{\operatorname{Var}(\E[U\mid Z])}{\beta^2}.
\]
The law of total variance implies $1\leq A\leq1+1/\alpha$ and $1\leq B\leq1+1/\beta$. Hence both converge to one and \eqref{eq:Gamma} converges to one. This completes the proof of Theorem~\ref{thm:sparsemultiplier}.

\section{Additional likelihood details}\label{app:likelihood}

This appendix records two observations useful for extensions.

\subsection{An explicit reconstruction kernel}

Fix a composition observation consisting of a named report vector $r$ and count $K$. Let
\[
 \mathcal B(r,K)=\{b\in\operatorname{supp}\Pi_N:
 K(b)=K,\ (r,b)\text{ lies in the common support}\}.
\]
For any $b\in\mathcal B(r,K)$, define
\[
 \mathcal K(b\mid r,K)=\Prb(b\mid r,K,\theta=+1).
\]
Equation \eqref{eq:exactLLR} implies that this conditional probability is unchanged if $\theta=+1$ is replaced by $\theta=-1$. Hence $\mathcal K$ is a state-independent Markov kernel. Define it arbitrarily when $\mathcal B(r,K)$ is empty. Applying it to $(r,K)$ reconstructs exactly the individually attributed experiment. This is the reverse garbling used in Theorem~\ref{thm:blackwell}.

\subsection{Hidden-disclosure accuracy at sparse fidelity}

If $Nx_N\to c$, a report agrees with the state under secrecy with probability $1/2+dx_N/2$. Let $P_N^+$ and $P_N^-$ denote the product report laws in the two states. The Bernoulli KL divergence per coordinate is $O(x_N^2)=O(N^{-2})$, so
\[
 \mathrm{KL}(P_N^+\Vert P_N^-)=O(N^{-1}).
\]
Pinsker's inequality gives $\lVert P_N^+-P_N^-\rVert_{\mathrm{TV}}=O(N^{-1/2})$. With equal state priors, optimal classification accuracy is $1/2+\lVert P_N^+-P_N^-\rVert_{\mathrm{TV}}/2$, proving the hidden-disclosure part of Theorem~\ref{thm:aggregation}.

{\small\raggedright

}

\end{document}